\documentclass[journal,10pt]{IEEEtran}

\IEEEoverridecommandlockouts
\usepackage{cite}
\usepackage{amsmath,amssymb,amsfonts,amsthm}
\usepackage{algorithmic}
\usepackage{graphicx}
\usepackage{textcomp}
\usepackage{xcolor}
\usepackage{array}
\usepackage{epstopdf}
\usepackage{pgfgantt} 
\usepackage{multicol}
\usepackage{lipsum}

\usepackage{array}
\def\BibTeX{{\rm B\kern-.05em{\sc i\kern-.025em b}\kern-.08em
    T\kern-.1667em\lower.7ex\hbox{E}\kern-.125emX}}
\ifCLASSOPTIONcompsoc
    \usepackage[caption=false, font=normalsize, labelfont=sf, textfont=sf]{subfig}
\else
\usepackage[caption=false, font=footnotesize]{subfig}
\fi

\newtheorem{theorem}{Theorem}

\newtheorem{lemma}[theorem]{Lemma}

\usepackage{array}
\newcolumntype{P}[1]{>{\centering\arraybackslash}p{#1}}
\newcolumntype{M}[1]{>{\centering\arraybackslash}m{#1}}

\def\bb0{{\mathbb{0}}}

\def\bb{{\mathbf{b}}}

\def\bff{{\mathbf{f}}}

\def\bh{{\mathbf{h}}}

\def\bx{{\mathbf{x}}}

\def\b0{{\mathbf{0}}}
\def\sfB{\mathsf{B}}

\def\sfN{\mathsf{N}}

\def\sfS{\mathsf{S}}
\def\sfT{\mathsf{T}}

\def\sfg{{\mathsf{g}}}

\def\sfj{{\mathsf{j}}}

\def\sfo{{\mathsf{o}}}

\def\sfr{{\mathsf{r}}}

\def\sf0{{\mathsf{0}}}

\def\rm0{{\mathrm{0}}}
\def\b0{{\pmb{0}}} 
\newcommand{\fdma}{\mathbf{f}_{\mathsf{dma}}}
\newcommand{\heff}{\mathbf{h}}
\newcommand{\hatt}{\mathbf{h}_{\mathsf{att}}}
\newcommand{\Rt}{R_{\mathsf{tune}}}
\newcommand{\Rtb}{B_{\mathsf{tune}}}
\newcommand{\Tt}{T_{\mathsf{temp}}}

\newcommand{\Aleak}{A_{\mathsf{leak}}}
\newcommand{\Wfill}{W_{\mathsf{fill}}}
\newcommand{\Fsel}{F_{\mathsf{freq}}}
\newcommand{\Wratio}{W_{\mathsf{ratio}}}

\newcommand{\Fcoup}{F_{\mathsf{coupl}}}
\newcommand{\Rres}{R_{\mathsf{res}}}

\newcommand{\hdma}{\mathbf{h}_{\mathsf{dma}}}

\newcommand{\am}{\alpha_{\mathsf{M}}}
\newcommand{\tam}{\bar{\alpha}_{\mathsf{M}}}
\newcommand{\wn}{w_{\nth}}

\newcommand{\btg}{\beta_{\mathsf{g}}}
\newcommand{\atg}{\bar{\alpha}_{\mathsf{g}}}
\newcommand{\tg}{\phi_{\mathsf{g}}}
\newcommand{\ft}{f_{\mathsf{t}}}
\newcommand{\Pin}{P_{\mathsf{in}}}

\newcommand{\Prad}{P_{\mathsf{rad}}}

\newcommand{\phio}{\varphi_{{\mathsf{o}}}}
\newcommand{\phit}{\phi_{\mathsf{t}}}

\newcommand{\Gp}{G_{{\mathsf{p}}}}

\newcommand{\Nt}{N_\mathsf{slot}}
\newcommand{\nth}{n}
\newcommand{\fr}{f_{\mathsf{r}}}
\newcommand{\dx}{d_{\mathsf{x}}}

\newcommand\figsizei{0.92}
\newcommand\figsizeii{0.92}
\newcommand{\blue}{}
\DeclareMathOperator*{\argmax}{argmax}
\DeclareMathOperator*{\argmin}{argmin}

\begin{document}

\title{Wideband dynamic metasurface antenna performance with practical design characteristics
}

\author{Joseph M. Carlson, \IEEEmembership{Student Member, IEEE}, Nitish V. Deshpande, \IEEEmembership{Student Member, IEEE}, \\ Miguel Rodrigo Castellanos, \IEEEmembership{Member, IEEE}, and Robert W. Heath, Jr., \IEEEmembership{Fellow, IEEE}

\thanks{Joseph M. Carlson, Nitish V. Deshpande, and Robert W.  Heath Jr. are with the 
	Department of Electrical and Computer Engineering, University of California San Diego, La Jolla, CA 92093 USA (e-mail: j4carlson@ucsd.edu; nideshpande@ucsd.edu; rwheathjr@ucsd.edu).
	Miguel Rodrigo Castellanos is with the Department of Electrical Engineering and Computer Science, University of Tennessee, Knoxville, TN 37996 USA (e-mail: mrcastellanos@utk.edu). This material is based upon work supported by the National Science Foundation under grant nos. NSF-ECCS-2435261, NSF-ECCS-2414678, NSF-CCF-2435254, and the Army Research Office under Grant W911NF2410107. A part of this paper has been accepted in the 2025 IEEE International Conference on Communications (ICC) \cite{j_icc}.}
}

\maketitle

\begin{abstract}
  Dynamic metasurface antennas (DMA) provide low-power beamforming through reconfigurable radiative slots. Each slot has a tunable component that consumes low power compared to typical analog components like phase shifters. This makes DMAs a potential candidate to minimize the power consumption of multiple-input multiple-output (MIMO) antenna arrays. In this paper, we investigate the use of DMAs in a wideband communication setting with practical DMA design characteristics. We develop approximations for the DMA beamforming gain that account for the effects of waveguide attenuation, element frequency-selectivity, and limited reconfigurability of the tunable components as a function of the signal bandwidth. The approximations allow for key insights into the wideband performance of DMAs in terms of different design variables. We develop a simple successive beamforming algorithm to improve the wideband performance of DMAs by sequentially configuring each DMA element. Simulation results for a line-of-sight (LOS) wideband system show the accuracy of the approximations with the simulated DMA model in terms of spectral efficiency. We also find that the proposed successive beamforming algorithm increases the overall spectral efficiency of the DMA-based wideband system compared with a baseline DMA beamforming method.
\end{abstract}


\section{Introduction}\label{sec: intro}

Dynamic metasurface antennas (DMAs) are a type of tunable leaky-wave antenna that gradually radiates power out through reconfigurable slots in a waveguide \cite{smith2017analysis}. \blue{DMAs differ from reflective or transmissive metasurfaces in that they operate as active antennas with a dedicated radio frequency (RF) feed, and can be integrated into a wireless transmitter or receiver architecture as an antenna array.} Low-power, low-cost tunable components are integrated into the DMA elements to control their resonant frequency and enable beamforming. Through the reconfigurable components, DMAs provide a potential solution to minimize both power consumption and hardware costs in multiple-input multiple-output (MIMO) systems when compared with phase shifter or digital beamforming methods \cite{shlezinger2021dynamic}. DMAs also have favorable wideband characteristics over static antennas due to the ability to tune the resonant frequency of each DMA element. Despite this, prior work has focused largely on narrowband communications for DMAs. In our work, we expand upon the current narrowband DMA models to evaluate the beamforming gain of DMAs with large signal bandwidths. 

\subsection{Prior work}

DMAs have been studied extensively in the past five years due to their potential power consumption and cost benefits over digital and phased arrays. The concept of metasurface antennas stems from using the electrically-small resonators in metasurface designs, such as split-ring-resonators \cite{wang2020metantenna}, as active antennas \cite{faenzi2019metasurface,badawe2016true}. The extension to dynamic metasurface antennas involves the integration of reconfigurable components into the metasurface antenna design to dynamically tune the antenna characteristics \cite{smith2017analysis,boyarsky2021electronically,SleasmanEtAlWaveguideFedTunableMetamaterialElement2016}. DMA designs have been proposed in literature using both varactor diodes \cite{boyarsky2021electronically} and p-i-n diodes \cite{lin2020high} to actively reconfigure the metasurface antenna with continuous or binary tuning capabilities. Both metasurface designs in \cite{boyarsky2021electronically} and \cite{lin2020high} were fabricated and experimentally measured to obtain steered beam patterns using the reconfigurable components. In essence, the experimental work shows that DMAs provide a novel and practical beamsteering architecture with reduced power consumption compared to digital and phased arrays, making them a potentially interesting technology for cellular systems.

There are several papers that have incorporated DMAs into more elaborate MIMO configurations. Simple beamforming strategies have been developed to map desired antenna weights to the Lorentzian-constrained weights \cite{smith2017analysis,SleasmanEtAlWaveguideFedTunableMetamaterialElement2016} and binary weights \cite{deng2022reconfigurable,lin2020high,hwang2020binary,pan2020design} to mimic a line-of-sight (LOS) beam pattern. More complex algorithms to configure the DMA Lorentzian weights have been developed to maximize spectral efficiency \cite{ShlezingerEtAlDynamicMetasurfaceAntennasUplink2019,WangJointTransmitterReceiverDesign2022,HuangEtAlStructuredOFDMModulationXLMIMO2025}, minimize bit-error rate in MIMO-orthogonal frequency division multiplexing (OFDM) systems \cite{WangEtAlDynamicMetasurfaceAntennasMIMOOFDM2021}, and maximize massive MIMO energy efficiency \cite{YouEtAlEnergyEfficiencyMaximizationMassive2022}. Extensions of the work in \cite{ShlezingerEtAlDynamicMetasurfaceAntennasUplink2019,WangJointTransmitterReceiverDesign2022,HuangEtAlStructuredOFDMModulationXLMIMO2025, WangEtAlDynamicMetasurfaceAntennasMIMOOFDM2021, YouEtAlEnergyEfficiencyMaximizationMassive2022} consider a hybrid precoding architecture with a DMA-based analog precoder and a digital precoder with limited RF chains \cite{ZhangEtAlBeamFocusingNearFieldMultiuser2022,KimaryoLeeDownlinkBeamformingDynamicMetasurface2023,AzarbahramEtAlEnergyBeamformingRFWireless2023,HuangEtAlJointMicrostripSelectionBeamforming2023}. Despite the Lorentzian weight constraint, DMA-based systems have been shown to provide spectral efficiency that approaches what can be achieved with typical phased arrays.

A common limitation of the work in \cite{ShlezingerEtAlDynamicMetasurfaceAntennasUplink2019,WangJointTransmitterReceiverDesign2022,YouEtAlEnergyEfficiencyMaximizationMassive2022,ZhangEtAlBeamFocusingNearFieldMultiuser2022,KimaryoLeeDownlinkBeamformingDynamicMetasurface2023,AzarbahramEtAlEnergyBeamformingRFWireless2023,HuangEtAlJointMicrostripSelectionBeamforming2023} is that the DMA element is modeled solely by the Lorentzian-constrained antenna weights, which does not account for all of the design features of DMAs that make them different from static antennas. Key aspects of the DMA model that are neglected in \cite{ShlezingerEtAlDynamicMetasurfaceAntennasUplink2019,WangJointTransmitterReceiverDesign2022,YouEtAlEnergyEfficiencyMaximizationMassive2022,ZhangEtAlBeamFocusingNearFieldMultiuser2022,KimaryoLeeDownlinkBeamformingDynamicMetasurface2023,AzarbahramEtAlEnergyBeamformingRFWireless2023,HuangEtAlJointMicrostripSelectionBeamforming2023} include the frequency-selectivity in the DMA element resonance, minimal impact of the waveguide attenuation as power is leaked out through the DMA elements, and the range of tunability for the DMA element resonant frequency. All three of these neglected assumptions will affect DMA beamforming, but the extent is unclear in current literature. Moreover, the work in \cite{ShlezingerEtAlDynamicMetasurfaceAntennasUplink2019,WangJointTransmitterReceiverDesign2022,YouEtAlEnergyEfficiencyMaximizationMassive2022,ZhangEtAlBeamFocusingNearFieldMultiuser2022,KimaryoLeeDownlinkBeamformingDynamicMetasurface2023,AzarbahramEtAlEnergyBeamformingRFWireless2023,HuangEtAlJointMicrostripSelectionBeamforming2023} largely assumes a narrowband system model, where the fractional bandwidth is small enough such that frequency-selective DMA effects are negligible. \blue{While work in \cite{HuangEtAlStructuredOFDMModulationXLMIMO2025} and \cite{WangEtAlDynamicMetasurfaceAntennasMIMOOFDM2021} includes the frequency-selective DMA resonance in the DMA model, the waveguide attenuation and resonant frequency tunability are not considered and could significantly impact the resulting spectral efficiency.} Since the DMA element resonance, waveguide attenuation, and resonant frequency tunability are all frequency-selective, increasing the signal bandwidth to consider a wideband wireless communication system has not been explored with these practical DMA design characteristics.

Design work on DMAs has provided a more complete characterization of their behavior. Foundational modeling of DMAs was done in \cite{Pulido-ManceraEtAlPolarizabilityExtractionComplementaryMetamaterial2017a,ScherKuesterExtractingBulkEffectiveParameters2009,pulido2018analytical_dissertation} to relate DMA elements to magnetic dipoles and show the frequency-selectivity resonant response of DMAs. The work in \cite{Pulido-ManceraEtAlPolarizabilityExtractionComplementaryMetamaterial2017a,ScherKuesterExtractingBulkEffectiveParameters2009,pulido2018analytical_dissertation}, however, does not explore the frequency-selective characteristics of DMAs in the context of communications and beamforming. Additional DMA models in \cite{j_asilomar_impedance,smith2017analysis,boyarsky2021electronically,boyarsky2020grating} were developed to account for waveguide attenuation and impedance in the context of beamforming, but did not provide a thorough analysis for how to design the DMA for a specific waveguide attenuation or how different attenuation values impact the resulting DMA beamforming. Moreover, all of the studies in \cite{Pulido-ManceraEtAlPolarizabilityExtractionComplementaryMetamaterial2017a,ScherKuesterExtractingBulkEffectiveParameters2009,pulido2018analytical_dissertation,smith2017analysis,boyarsky2021electronically,boyarsky2020grating} do not discuss the impact of the DMA resonant frequency tuning abilities on DMA beamforming. In a conference version of this paper, we include the resonant frequency tuning, waveguide attenuation and frequency-selectivity in the DMA model, but do not analyze how these parameters impact the DMA beamforming gain for large bandwidths \cite{j_icc}. In our paper, we will close this gap and analyze the affects of the DMA element resonance, waveguide attenuation, and resonant frequency tunability on the DMA beamforming gain for large signal bandwidths, where the DMA frequency-selectivity cannot be neglected. 

\blue{To highlight the challenges addressed in our paper, it is important to understand the impact of DMA element frequency-selectivity, waveguide attenuation, and resonant frequency tunability on the DMA wideband beamforming gain. For example, the frequency-selectivity of the DMA element resonant response causes the amplitude and phase of the effective DMA beamforming weights to vary drastically across frequency, and the waveguide attenuation applies a frequency-selective amplitude taper across all the DMA beamforming weights. The resonant frequency tunability also places limits on the possible DMA beamforming weights, which adds to the difficulty in obtaining a tractable analysis of the DMA wideband beamforming gain. Overall, the general relationship for DMA wideband beamforming gain as a function of frequency is not clear, and configuring the DMA beamforming weights to improve wideband beamforming gain with these practical design characteristics remains an open challenge. We address these challenges in our paper by deriving a DMA wideband beamforming gain approximation based on the impact of the frequency-selective DMA resonance, waveguide attenuation, and resonant frequency tunability, and we propose a successive beamforming algorithm to improve the DMA wideband spectral efficiency under the practical design characteristics.}

\subsection{Contributions}

We analyze the impact of DMA characteristics, such as resonant frequency tunability and waveguide attenuation, on the beamforming gain and spectral efficiency in a multiple-input single-output (MISO) DMA-analog wireless communication system with large signal bandwidths. The key contributions in this paper are summarized as follows. We develop a wideband DMA signal model that includes the waveguide attenuation, resonant frequency tunability, and DMA resonance frequency-selectivity. These three DMA design characteristics have a significant impact on the resulting DMA beamforming gain as the signal bandwidth increases, yet there has been no comprehensive analysis in prior work that integrates these DMA characteristics into a signal model. Next, we derive an approximation of the beamforming gain for a wideband MISO DMA system under the influence of the waveguide attenuation, limited resonant frequency tunability, and frequency-selective elements. The approximation provides a closed-form expression for the resulting beamforming gain of the DMA as a function of the signal bandwidth and DMA design parameters, such as the coupling and damping factors, as well as the DMA array spacing. Therefore, the design parameters for any physical DMA can be applied to the approximation to obtain an accurate estimate on the achievable beamforming gain and spectral efficiency in a DMA-based wireless system.

Using the derived approximation, we gain valuable insights into the best DMA design characteristics to maximize the MISO spectral efficiency for small and large signal bandwidths. We show that the limited resonant frequency tunability decreases spectral efficiency by reducing the available beamforming weights, meaning the range of resonant frequency tunability is a crucial DMA design characteristic for DMA-based wireless systems. \blue{We also demonstrate the need for small antenna element spacing to minimize the impact of high frequency-selectivity due to the wireless and waveguide channels, as well as the DMA resonance frequency-selectivity. Moreover, in simulation results we show that large damping factors generally lead to higher data rates since there is less frequency-selectivity in the DMA element response.} Finally, we propose a successive beamforming algorithm to configure the DMA elements. The proposed successive beamforming algorithm accounts for both the frequency-selective DMA design characteristics and the configurations of other DMA elements to improve the wideband spectral efficiency. We find that the proposed algorithm provides an increase in spectral efficiency and data rates when compared with a baseline DMA beamforming approach, especially as the signal bandwidth becomes large. \blue{We also demonstrate that the proposed successive beamforming algorithm outperforms the baseline beamforming algorithm and hybrid beamforming architectures when extended to a multipath channel environment, which shows that the proposed algorithm functions well in complex channel scenarios beyond the initial LOS case.}

{\textit{Organization}}: In Section \ref{sec: dma model}, we outline the wideband MISO-OFDM signal model. In Section \ref{sec: DMA channel and beamformer}, we define the frequency-selective DMA model and waveguide propagation model. We also introduce a novel DMA parameter called the tuning bandwidth to account for the ability to tune the DMA resonant frequency. In Section \ref{sec: dma approx}, we derive a beamforming gain approximation to analyze spectral efficiency. The beamforming gain approximation includes the effects of the DMA element frequency-selectivity, tuning bandwidth, and waveguide attenuation unlike prior work in \cite{ShlezingerEtAlDynamicMetasurfaceAntennasUplink2019,WangJointTransmitterReceiverDesign2022,YouEtAlEnergyEfficiencyMaximizationMassive2022,ZhangEtAlBeamFocusingNearFieldMultiuser2022,KimaryoLeeDownlinkBeamformingDynamicMetasurface2023,AzarbahramEtAlEnergyBeamformingRFWireless2023,HuangEtAlJointMicrostripSelectionBeamforming2023}. Next, in Section \ref{sec: successive} we define a baseline DMA algorithm for configuring the antenna weights based on \cite{smith2017analysis} and propose a successive algorithm to increase the DMA wideband spectral efficiency. Unlike \cite{smith2017analysis,boyarsky2021electronically}, the proposed beamforming algorithm is frequency-selective and maximizes the spectral efficiency across a multi-carrier system, rather than simply at a center frequency. In Section \ref{sec: results}, we then show simulation results to compare the proposed successive algorithm, baseline DMA algorithm and beamforming gain approximation. The simulation results extend the DMA modeling work developed in \cite{Pulido-ManceraEtAlPolarizabilityExtractionComplementaryMetamaterial2017a,ScherKuesterExtractingBulkEffectiveParameters2009,pulido2018analytical_dissertation} to a wireless setting, where we convey how different DMA design parameters can be optimized to improve spectral efficiency.
Lastly, in Section \ref{sec: conclusion} we provide a summary of the paper and future research directions.

{\textit{Notation}}: We denote a bold, capital letter $ \mathbf{A} $ as a matrix, a bold, lowercase letter $ \mathbf{a} $ as a vector, and a script letter $ \mathcal{A} $ as a set. Let $\mathbf{A}^*$ represent the matrix conjugate transpose, $\mathbf{A}^{\mathsf{c}}$ represent the matrix conjugate, and $\mathbf{A}^{\mathsf{T}}$ represent the matrix transpose. For a complex number $a$, we use $\operatorname{Re}(a)$ and $\operatorname{Im}(a)$ to indicate its real and imaginary part, $|a|$ as its magnitude, and $\angle a$ as its phase. We define the operator $\odot$ as the Hadamard product. We denote $||\mathbf{a}||_\mathsf{F}$ as the Frobenius norm of vector $\mathbf{a}$. We use $O(\cdot)$ to denote the big O notation.

\section{System model}\label{sec: dma model}

We first establish a frequency-selective MISO-OFDM signal model for a DMA-based wireless system. We also define the key metrics for evaluating the DMA performance as the wideband beamforming gain and spectral efficiency.

\subsection{MISO-OFDM signal model}

The system consists of a uniform linear array (ULA) DMA transmitter with $\Nt$ elements communicating with a single-antenna user through a wideband MISO-OFDM channel. Fig. \ref{fig: dma_array} shows the DMA ULA architecture with $\Nt$ radiating elements atop a waveguide. Our goal is to analyze the effects of the frequency-selective DMA elements on the overall MISO-OFDM data rates and spectral efficiency. We consider a DMA model with continuous tunable components, such as varactors, that are configured to beamform by adjusting the components to create a resonant response in the DMA element.

We now define the MISO-OFDM signal model. The DMA transmitter communicates with the user over $K$ subcarriers by precoding the $k$th subcarrier transmit symbol $s[k]$, which has zero-mean unit-variance.
Let $\fdma[k] \in \mathbb{C}^{\Nt \times 1}$ be the DMA transmit beamforming vector for the $k$th subcarrier and $\mathbf{h}_\mathsf{att}[k] \in \mathbb{R}^{\Nt \times 1}$ be an amplitude taper on the DMA beamforming weights based on the leakage of radiated power. This is discussed in further detail in Section \ref{sec: channel}. The  transmit power input to the DMA is $\Pin$ for all subcarriers, such that for a total system input power $P_{\mathsf{in,tot}}$, the per-subcarrier input power is $\Pin = \frac{P_{\mathsf{in,tot}}}{K}$.  Let $G_{\mathsf{dma}}$ be a DMA loss term that accounts for the DMA efficiency, and let $M_k$ be a normalization term that satisfies a power constraint on the DMA beamformer.
The transmit signal vector is denoted as 
\begin{align}
	\bx[k]= \sqrt{\Pin  G_{\mathsf{dma}} M_k} (\fdma[k] \odot \mathbf{h}_\mathsf{att}[k]) s[k].
\end{align}
 We assume Gaussian-distributed noise $n[k] \sim \mathcal{N}_{\mathbb{C}}(0,\sigma_\mathsf{N}^2)$.  
 
 Next, we define the received signal model for the DMA. Let $\heff[k] \in \mathbb{C}^{\Nt \times 1}$ be the small-scale  fading channel between the DMA and user receive antenna. Let $G_{\mathsf{P},k}$ be the term which captures the frequency dependent large-scale fading path loss, which is further defined in Section \ref{subsec: snr and se}. The MISO-OFDM single-stream received signal for a single time slot  is
\begin{equation}
    y[k] = \sqrt{G_{\mathsf{P},k}} \mathbf{h}^{\mathsf{T}}[k]\bx[k] + n[k]. 
\end{equation}
The characteristics of the DMA element are incorporated into the beamformer $\fdma[k]$, while properties of the waveguide are integrated into the channel $\heff[k]$ and the leakage vector $\hatt[k]$.

\begin{figure}
	\centering
	\includegraphics[width=\figsizei\linewidth]{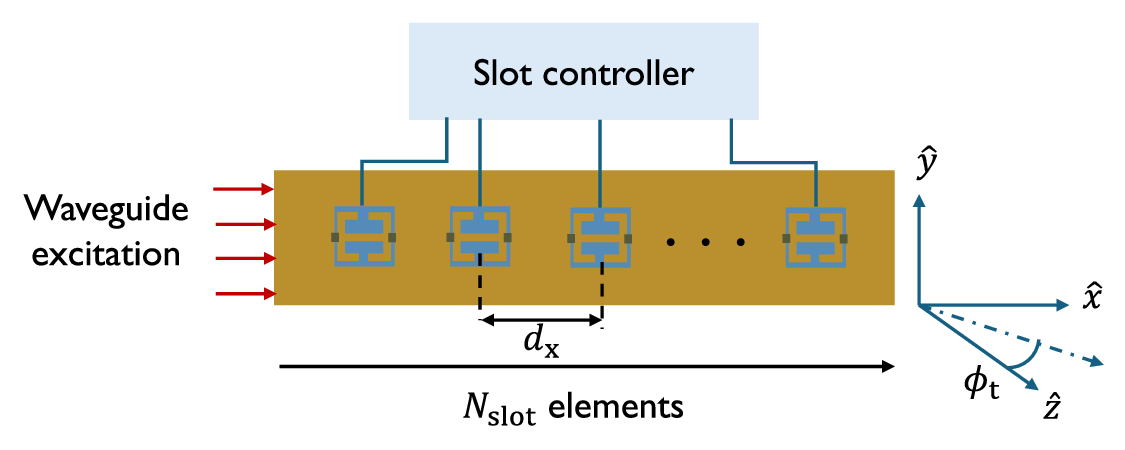}
	\caption{\blue{Architecture of a DMA as a leaky-wave uniform linear array of antenna elements. Varactor diodes are integrated into each antenna element to enable low-power beamforming with the use of an external controller.}}
	\label{fig: dma_array}
\end{figure}

\subsection{Wideband beamforming gain and spectral efficiency}\label{subsec: snr and se}

We now describe the signal-to-noise ratio (SNR)  model that incorporates the signal bandwidth of the DMA wideband system, and the performance metrics as the wideband beamforming gain and spectral efficiency. The large-scale fading path loss is assumed to be free space path loss and expressed as $G_{\mathsf{P},k}= \left(\frac{c}{f_k 4 \pi r}\right)^2$  for a distance $r$ from the DMA. Let $\sigma_{\sfN}^2 = k_{\sfB}\Tt B/K$ be the noise power for the Boltzmann constant $k_{\sfB}$, temperature $\Tt$, and subcarrier bandwidth $B/K$. Then, we define the per-subcarrier SNR with the DMA as
\begin{equation}\label{eq: snr}
	\rho_k = G_{\mathsf{P},k}G_{\mathsf{dma}}\frac{\Pin }{\sigma^2_{\sfN}}.
\end{equation}
We see that the SNR in \eqref{eq: snr} depends on the subcarrier bandwidth $B/K$ within the noise power term $\sigma_{\sfN}^2$, which we will examine in Section \ref{sec: successive results}.

Next we describe the key performance metrics to evaluate the DMA beamformer. We define the sum beamforming gain for the system across all subcarriers as
\begin{equation}
    G_\mathsf{sum} = \sum\limits_{k=1}^{K} M_k|\heff^{\mathsf{T}}[k] \left(\fdma[k] \odot \mathbf{h}_{\mathsf{att}}[k] \right)|^2
\end{equation}
and the spectral efficiency as
\begin{multline}\label{eq: spec eff}
	C\Big(\{ \rho_k, \heff[k]\}_{k=1}^K \Big) = \\ \frac{1}{K} \sum\limits_{k=1}^K \log_2 \left( 1+ \rho_k M_k |\heff^{\mathsf{T}}[k] \left(\fdma[k] \odot \mathbf{h}_{\mathsf{att}}[k] \right)|^2 \right).
\end{multline}
The resulting sum beamforming gain and spectral efficiency depend on a multitude of DMA design parameters: the resonant frequency tuning bandwidth, waveguide power leakage, element frequency-selectivity, and system parameters like the signal bandwidth. We will derive an expression to approximate the DMA sum beamforming gain, and extend the approximation to investigate how spectral efficiency changes as a function of the DMA design parameters.

\section{DMA channel and beamformer model}\label{sec: DMA channel and beamformer}

We now define the LOS DMA channel model that incorporates the DMA waveguide propagation effects and the wireless channel. We also model the DMA beamformer by the reconfigurable magnetic dipole response of each element, and derive a frequency-selective Lorentzian-constraint to extend the prior DMA weight model in \cite{smith2017analysis} to a wideband setting.

\subsection{Channel and waveguide model}\label{sec: channel}

The channel between the transmitter and the receiver includes the effects of over-the-air propagation and the DMA waveguide. We assume a LOS channel model to simplify the wideband system analysis. \blue{Let $f$ be the signal frequency}, let $f_k$ be the $k$th subcarrier passband frequency, let $c$ indicate the speed of light in a vacuum, let $\dx$ be the inter-element spacing, and let $\phit$ be the LOS direction of the user from broadside, as shown in Fig. \ref{fig: dma_array}. We assume the DMA is oriented in the $xy$ plane and the DMA elements are spaced along the $x$ direction. The transmit array response vector is
\begin{equation}\label{eq: steering vec}
	\mathbf{a}[\phit, k] = \left[e^{\sfj0 \left(\frac{2\pi f_k}{c}\right) \dx \sin \phit}, \ldots, e^{\sfj (\Nt-1) \left(\frac{2\pi f_k}{c}\right) \dx \sin \phit}\right]^{\sfT}.
\end{equation}
In addition to the transmit array vector, the DMA waveguide also produces attenuation and phase shift effects on the DMA elements due to the propagation of the transmit signal through the waveguide. 

Since DMAs act as reconfigurable leaky-wave antennas, the feed for the individual DMA elements is the waveguide. The electromagnetic fields attenuate along the waveguide as power is radiated out through the elements. There is an inherent phase advance that is introduced for each DMA element from the propagation of the electromagnetic fields. We denote the  waveguide phase constant as $\beta_{\mathsf{g}}(f)$.
	Similar to \cite{smith2017analysis}, we ignore the resistive loss and only model the  waveguide attenuation due to radiated power leakage which is denoted as $\bar{\alpha}_{\mathsf{g}}(f)$.  The subscript $\sfg$ represents the waveguide. Assuming zero phase angle at the first element,  we define the DMA waveguide channel that incorporates the phase advance as \cite{smith2017analysis}
\begin{equation}\label{eq: dma waveguide channel}
	\hdma[k] = \left[ e^{-\sfj (0) \dx  \btg(f_k) }, \ldots, e^{- \sfj(\Nt-1) \dx  \btg(f_k)} \right]^\sfT,
\end{equation}
The total wireless channel is obtained by combining the array vector $\mathbf{a}[\phit,k]$ with the effects of the DMA waveguide channel $\hdma[k]$ as
\begin{equation}\label{eq: eff channel}
	\heff[k] = \mathbf{a}[\phit,k] \odot \hdma[k].
\end{equation}
The expression in \eqref{eq: eff channel} describes the total array response vector including both wireless and waveguide channel propagation.

Additionally, we define a waveguide attenuation vector that describes the leakage of radiated power along the length of the waveguide. As the specific configuration of DMA elements impacts the leakage of radiated power, we define a leakage constant $\bar{\alpha}_{\mathsf{g}}(f)$ that approximates the attenuation of the waveguide fields as power is radiated out through the DMA elements. We provide further justification of using an approximate leakage constant in Section \ref{sec: results}. We define the leakage vector as
\begin{equation}\label{eq: dma atten vector}
	\mathbf{h}_\mathsf{att}[k] = \left[ e^{-(0) \dx \atg(f_k) }, \ldots, e^{-(\Nt-1) \dx \atg(f_k)} \right]^\sfT.
\end{equation}
We provide separate expressions for the waveguide effects of phase propagation and leakage because of their impact on the DMA beamformer. The DMA phase propagation channel in \eqref{eq: dma waveguide channel} contributes a phase advance to the wireless channel that must be accounted for and counteracted by the DMA beamformer to steer a beam in a desired direction. The DMA leakage vector in \eqref{eq: dma atten vector}, however, applies an effective amplitude taper on the resulting DMA beamforming weights $\fdma[k]$. For now, we focus on configuring the DMA beamforming weights prior to incorporating the amplitude taper from the leakage vector in \eqref{eq: dma atten vector} for simplicity. Later on, we propose a novel method to configure the DMA beamforming weights in Section \ref{subsec: successive alg} which accounts for the amplitude taper.

\subsection{Frequency-selective DMA beamforming  model}\label{subsec: mag pol}

We establish a model for the tunable responses of the DMA elements and incorporate additional physical DMA effects into the DMA beamforming model. 
The polarizable dipole framework from \cite{smith2017analysis} is used to model the DMA  beamforming similar to prior work\cite{carlson2023hierarchical, nitish_qif1}.
The radiation properties of each DMA slot resembles
that of magnetic dipoles as was experimentally verified in \cite{yoo2022experimental}. The radiated field is governed by the  magnetic polarizability of each dipole\cite{smith2017analysis}. Following the notation from \cite{smith2017analysis},  let $\Fcoup$ be the coupling factor, $\Gamma$ be the damping factor, and $f_{\sfr, \nth}$ be the tunable element resonant frequency of the $\nth$th slot. The DMA magnetic polarizability as a function of the subcarrier frequency and  the $\nth$th element resonant frequency is \cite{smith2017analysis}
\begin{equation}\label{eq: mag pol}
    \am(f_k, f_{\sfr, \nth}) = \frac{2\pi f_k^2 \Fcoup}{2\pi f_{\sfr,\nth}^2-2\pi f_k^2+\sfj\Gamma f_k}.
\end{equation}
Since \eqref{eq: mag pol} depends on both $ f_{\sfr, \nth}$ and $f_k$, the frequency-selective resonant response is dependent on the resonant frequency and the signal bandwidth $B$ that contains the subcarrier frequencies $f_k$.

Next, we define a tuning bandwidth term to model the physical limitations and tuning capabilities of the reconfigurable varactors. For a minimum resonant frequency $f_{\sfr,{\mathsf{min}}}$ and maximum resonant frequency $f_{\sfr,{\mathsf{max}}}$, the resonant frequency tuning range is $\Rt \in [f_{\sfr,{\mathsf{min}}},f_{\sfr,{\mathsf{max}}}]$. We define the tuning bandwidth as
\begin{equation}
    \Rtb = f_{\sfr,{\mathsf{max}}}-f_{\sfr,{\mathsf{min}}}
\end{equation}
which represents the bandwidth of possible resonant frequency tuning values $\fr$ around the center operating frequency $\ft$. 

The tuning bandwidth $\Rtb$ has a large impact on the resulting spectral efficiency as it limits the flexibility of the DMA elements to create different antenna weights. Several design parameters affect the tuning bandwidth: DMA element geometry and varactor diode parameters. We discuss this further in Section \ref{sec: dma approx}. In this paper, we assume the tuning bandwidth can be adjusted based on a DMA design, such as through altering the DMA geometry or selecting different varactor diodes \cite{lin2020high}. We will analyze the importance of large and small tuning bandwidths on the DMA wideband performance.

The tunable element responses allow the DMA to achieve beamforming based on the tuning configuration. The magnetic polarizability then becomes the effective antenna weight.
We define the quality factor as $Q_k = \frac{2\pi f_{k}}{ \Gamma}$ and the normalized magnetic polarizability response as 
\begin{equation}\label{eq: mag pol norm}
    \tam(f_k,f_{\sfr, \nth}) = \frac{1}{Q_k F}\am(f_k,f_{\sfr, \nth}).
\end{equation}
The factor $Q_k F$ ensures that the highest magnitude beamforming weight possible is one, which maintains consistency with the Lorentzian model for the DMA weights. We can then leverage the beamforming methods for DMAs using the Lorentzian model to map unit-amplitude weights onto a frequency-selective Lorentzian-constraint. The transmit beamforming vector for the DMA based on the  varactor tunings per element is
%$M = \frac{2\pi f_k F}{\Gamma}$ is
\begin{equation}
    \fdma[k] = [\tam(f_k,f_{\sfr, 0}), \ldots, \tam(f_k,f_{\sfr,\Nt-1})]^{\sfT}.
\end{equation}
We use the normalization term $M_k$ to satisfy a radiated power constraint on the DMA beamformer, as discussed in Section \ref{sec: results}.

\subsection{Frequency-selective Lorentzian constraint}
Next, we define the frequency-selective Lorentzian constraint to describe the DMA beamforming weights $\tam$. Presenting the frequency-selective magnetic polarizability term from \eqref{eq: mag pol} in terms of a frequency-selective Lorentzian constraint allows for a simplified analysis of the DMA beamformer and its wideband beamforming gain performance. For a phase angle $\zeta$ representing the reconfigurability of the DMA element, the Lorentzian constraint is \cite{smith2017analysis}
\begin{equation}\label{eq: lor con}
    \mathcal{Q} = \left\{-\frac{\sfj-e^{\sfj\zeta}}{2} : \zeta\in[0,2\pi]\right\}.
\end{equation}
\blue{Fig. \ref{fig: lorentz_constraint} shows a plot of the Lorentzian weight distribution in \eqref{eq: lor con} compared with a unit-amplitude weight distribution as $\mathcal{Q} = \left\{ e^{\sfj\zeta} : \zeta\in[0,2\pi] \right\}$. It is important to note that the Lorentzian weight distribution restricts the available phase range for beamforming to $[-\pi,0]$, whereas the unit-amplitude weight distribution allows for the full $[0,2\pi]$ phase range.}

Based on the Lorentzian constraint in \eqref{eq: lor con}, to derive the frequency-selective magnetic polarizability term, let 
% Replacing \vartheta by \Psi
\begin{equation}\label{eq: mag pol angle}
    \Psi(f,f_{\sfr, \nth}) = \angle \tam(f,f_{\sfr, \nth}) = \arctan \left( \frac{2\pi (f_{\sfr, \nth}^2-f^2)}{\Gamma f} \right)
\end{equation} represent the angle of the magnetic polarizability. Then, we can express the magnetic polarizability in \eqref{eq: mag pol norm} using the Lorentzian form in \eqref{eq: lor con} as
\begin{equation}\label{eq: mag pol cos}
    \tam(f,f_{\sfr, \nth}) = \cos(\Psi(f,f_{\sfr, \nth})e^{\sfj(\Psi(f,f_{\sfr, \nth})-\frac{\pi}{2})}.
\end{equation}
We relate \eqref{eq: mag pol cos} to the Lorentzian constraint in \eqref{eq: lor con} as follows.

\begin{figure}
	\centering
	\includegraphics[width=0.9\linewidth]{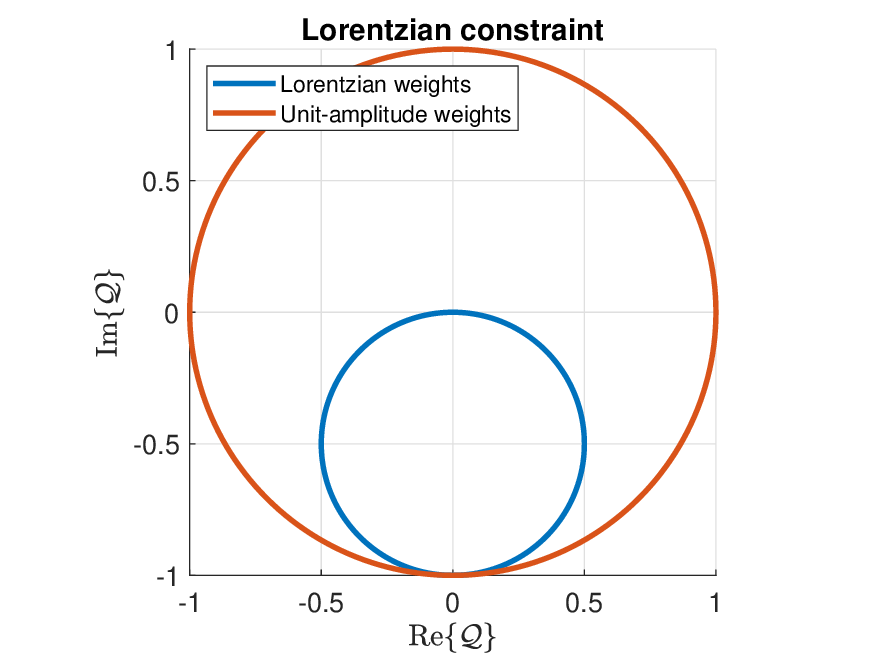}
	\caption{\blue{Lorentzian-constrained weight distribution for the DMA elements compared with unit-amplitude weights. We find that the Lorentzian-constrained weights allow for a total phase range of $[-\pi,0]$.}}
	\label{fig: lorentz_constraint}
\end{figure}

\begin{lemma}\label{lem:  mag pol}
The magnetic polarizability term $\am$ in \eqref{eq: mag pol norm} is related to the Lorentzian constraint in \eqref{eq: lor con} by 
\begin{equation}\label{eq: mag pol lor}
    \tam(f,f_{\sfr, \nth}) = -\frac{\sfj-e^{\sfj 2  \Psi(f,f_{\sfr, \nth})}}{2}.
\end{equation}
\end{lemma}
\begin{proof}
	Check Appendix~\ref{proof:  lem mag pol} for the proof.
\end{proof}
\noindent \blue{The magnetic polarizability expression in \eqref{eq: mag pol lor} provides a useful reformulation of the original magnetic polarizability expression in \eqref{eq: mag pol} to match the Lorentzian constraint in \eqref{eq: lor con}. While \eqref{eq: mag pol lor} does not provide a solution for optimizing the resonant frequencies, the Lorentzian form of the frequency-selective magnetic polarizability in \eqref{eq: mag pol lor} helps to simplify the analysis for developing a wideband beamforming gain approximation.}

Next, we use the magnetic polarizability in \eqref{eq: mag pol lor} to derive an approximation on the beamforming gain of the wideband DMA system. The DMA beamformer takes the form of \eqref{eq: mag pol lor} as
\begin{equation}\label{eq: dma beamformer}
    \fdma[k] = \left[-\frac{\sfj-e^{\sfj 2  \Psi(f_k,f_{\sfr, 0})}}{2}, \ldots, -\frac{\sfj-e^{\sfj 2  \Psi(f_k,f_{\sfr, \Nt-1})}}{2}\right]^{\sfT}.
\end{equation}
We derive an approximation for the beamforming gain based on the formulation of the DMA beamformer in \eqref{eq: dma beamformer} to simplify the beamforming gain analysis as follows. We also denote $\mathbf{\Psi}(f_k) = [\Psi(f_k,f_{\sfr, 0}), \ldots, \Psi(f_k,f_{\sfr, \Nt-1})]^\sfT$ as the vector of frequency-selective DMA phase weights to simplify the notation.
\begin{lemma}\label{lem: 2}
For large $\Nt$ and either $\phit \neq \phi_{\mathsf{g}}$ or $\btg > \frac{2\pi f}{c}$, the DMA beamforming gain is approximately
\begin{equation}\label{eq: dma bf gain}
    |\heff^{\mathsf{T}}[k] \fdma[k]|^2 \approx \left|\frac{1}{2}\heff^{\mathsf{T}}[k]e^{\sfj 2 \mathbf{\Psi}(f_k)} \right|^2.
\end{equation}

\end{lemma}
\begin{proof}
	Check Appendix~\ref{proof of lemma 2} for the proof.
\end{proof}

 \noindent In Section \ref{sec: dma approx}, we develop methods to account for the impact of frequency-selective, tuning bandwidth and waveguide attenuation on the DMA beamforming gain based on \eqref{eq: dma bf gain}.

\subsection{Approximation of phase of the polarizability expression and center-frequency DMA beamforming}

We now provide an approximation for the magnetic polarizability function to further simplify the analysis of the DMA beamforming gain. The following lemma gives an approximation for the polarizability phase using a first-order Taylor series expansion.

\begin{lemma}\label{lem: 3}
A linear approximation for the phase of the magnetic polarizability is
\begin{equation}\label{eq: ang pol angle approx}
   \Psi(f,f_{\sfr, \nth}) = -\frac{\pi}{2} - \frac{4\pi}{\Gamma}(f-f_{\sfr, \nth}) + O(f^2).
\end{equation}
\end{lemma}
\begin{proof}
	Check Appendix~\ref{proof: lem 3} for the proof.
\end{proof}
\noindent In  Section \ref{sec: dma approx}, we use \eqref{eq: ang pol angle approx} to determine the impact of frequency-selectivity, tuning bandwidth, and waveguide attenuation on the wideband MISO beamforming gain.

We now define a simple DMA beamforming method for configuring the DMA weights based on the Lorentzian mapping technique defined in \cite{smith2017analysis} that configures the DMA elements at the center frequency of the bandwidth. Since the beamformer does not attempt to incorporate the frequency-selective channel or DMA effects, we consider this beamformer as a performance baseline. We first define the phase of the frequency-selective effective wireless channel. Based on the effective channel model in \eqref{eq: eff channel}, we define the phase of the $\nth$th effective channel element $h_\nth = e^{\sfj \varphi_{\nth,k}}$ for a substrate permittivity $\epsilon_{\mathsf{r}}$ and waveguide cutoff frequency $f_{{\mathsf{c}},10}$ as 
\begin{equation}\label{eq: channel phase}
    \varphi_{\nth,k} = \dx(\nth-1) \left( \frac{2 \pi f_k}{c}  \sin\phit + \frac{2\pi \epsilon_{\mathsf{r}}}{c} \sqrt{f_k^2-f_{{\mathsf{c}},10}^2} \right)
\end{equation}
The phase term in \eqref{eq: channel phase} includes both the wireless and waveguide propagation, where the propagation phase constant $\btg$ has been decomposed \cite{balanis2016antenna}. We will account for the effects of waveguide attenuation in Section \ref{sec: atten}.

Next, we use \eqref{eq: ang pol angle approx} to match the magnetic polarizability angle to the angle of the effective channel at the center frequency. We consider only the $e^{\sfj 2 \Psi}$ term in \eqref{eq: mag pol lor} as it accounts for the reconfigurability of the DMA elements for beamforming. Substituting \eqref{eq: ang pol angle approx} into \eqref{eq: mag pol lor} as the approximate Lorentzian beamformer, we find that a tuning bandwidth of $\Rtb = \frac{\Gamma}{4}$ results in a phase tunability over the range $[-\pi,0]$ at the center frequency. We assume this tuning bandwidth to define the DMA beamformer at the center frequency $f_{ \frac{K}{2}}$, and aim to counteract the channel phase $\varphi_{\nth,k}$ with
\begin{equation}\label{eq: center freq tuning}
    f_{\sfr, \nth}=   -\frac{\Gamma}{8\pi}\varphi_{\nth, \frac{K}{2} }  +   f_{ \frac{K}{2}} .
\end{equation}
Setting the resonant frequencies of the DMA elements according to \eqref{eq: center freq tuning} will guarantee that the phase of the effective wireless channel aligns with the DMA beamformer at the center frequency of the signal bandwidth. In the next section, we will use the center frequency DMA beamforming method based on \eqref{eq: center freq tuning} to derive an approximation for the wideband DMA beamforming gain.

\section{Beamforming gain approximation for a wideband DMA-based communication system}\label{sec: dma approx}

In this section, we derive an approximation for the beamforming gain of a wideband DMA system with a LOS wireless channel. We investigate the impact of the DMA frequency selectivity, tuning bandwidth and waveguide power leakage on the resulting data rates as a function of the signal bandwidth $B$ and use the approximation to give valuable insights into the wideband performance of DMAs. We develop the approximation by deriving separate expressions for the DMA frequency selectivity, tuning bandwidth and waveguide power leakage, and combining these terms heuristically to create the full approximation for the DMA beamforming gain. \blue{We assume an LOS wireless channel to allow for the tractable analysis of the DMA wideband beamforming gain, where a multipath channel environment would significantly complicate the development of a closed-form expression. Therefore, we focus on the LOS case now, and include simulation results in Section \ref{subsec: multipath} to demonstrate the effectiveness of DMAs in a wideband system with multipath propagation.}

\blue{Due to the non-convex frequency-selective Lorentzian constraint, it is very difficult to derive a closed-form solution for optimizing the design variables. Moreover, optimizing the design variables for a specific set of resonant frequencies is not practical, since the resonant frequencies need to be changed to enable beamforming while the design variables are a part of the physical DMA design that cannot be changed once the DMA is fabricated. Therefore, we focus on deriving an approximation for the beamforming gain of a wideband DMA system and extract important insights from the beamforming gain expression to propose broad design guidelines for the DMA elements.}

 To define the DMA beamforming gain approximation, let $\Fsel(f_k)$ be a term accounting for the frequency-selectivity of the beamforming gain.
 In addition to the term $\Fsel(f_k)$, to account for the losses due to practical issues like waveguide attenuation and limited tunability of the DMA, we define two penalty terms.
 We denote a DMA angular weight fill as $\xi \in [0,\pi)$ which is defined in Section~\ref{sec: Impact of tuning bandwidth } as the range of available phase values for the DMA weights. For a radiated power $\Prad$, we also denote a fractional radiated power term as $\Lambda = \frac{\Prad}{\Pin}$, which is discussed further in Section \ref{subsec: Design considerations}. Let $\Wfill(\xi) \in [0,1]$ be the penalty term which accounts for the impact of the limited tuning bandwidth. Let $\Aleak(\Lambda) \in [0,1]$  be the term which accounts for the impact of waveguide attenuation on beamforming gain. With these terms, we propose a heuristic approximation for the wideband DMA beamforming gain as 
 \begin{equation}\label{eq: approx bound}
	|\heff^{\mathsf{T}}[k] \left(\fdma[k] \odot \mathbf{h}_{\mathsf{att}}[k] \right)|^2 \approx \Fsel(f_k)\Wfill(\xi) \Aleak(\Lambda).
\end{equation}
Our goal is to use the approximation to extract important insights into the performance of a wideband DMA system. By separating the beamforming gain terms, we can determine how each term is affected by the DMA design parameters like the damping factor, steering angle, or propagation constant. We describe the methodology and derive expressions for each of these terms in the following sections.

\subsection{Impact of frequency selectivity}\label{sec: freq sel}

We now derive an expression to characterize the beamforming gain of the DMA across the entire bandwidth. Using \eqref{eq: center freq tuning}, the MISO beamforming gain with the center-frequency beamformer (without waveguide attenuation) is
\begin{multline}\label{eq: bf gain freq}
   \Fsel(f_k) = |\heff^{\mathsf{T}}[k] \fdma[k]|^2 
    = \\  \left| \frac{1}{2\sqrt{\Nt}}\sum\limits_{\nth=1}^{\Nt }
 \exp\left({\sfj \varphi_{\nth,k}-\sfj \varphi_{\nth, \frac{K}{2}} }\right) \right|^2.
\end{multline}
The DMA at the center frequency $f_{\frac{K}{2}}$ achieves the maximum beamforming gain, but there is a phase mismatch at other frequencies in the bandwidth that will degrade the overall beamforming gain. This effect is well-studied in the context of phased arrays and is known as beam-squint~\cite{10002944}. DMAs, however, suffer greater losses in wideband settings compared to phased arrays because of the combination of wireless and waveguide channel frequency-selectivity. We quantify the frequency-selective beamforming gain loss from beam-squint at each subcarrier frequency as follows.  

Let $\Delta_k = f_k-f_{\frac{K}{2}}$ be the difference between a subcarrier frequency and the center frequency of the bandwidth. The phase difference between the subcarrier frequency $f_k$ and center frequency $f_{\frac{K}{2}}$ of the $\nth$th element due to the effective channel is then
\begin{multline}\label{eq: chin}
	\chi_{\sfo}[k] = - \dx \Biggl( \frac{2 \pi \Delta_k}{c} \sin\phit  \\ + \frac{2\pi \epsilon_{\mathsf{r}}}{c} \left( \sqrt{(f_{\frac{K}{2}}+\Delta_k)^2-f_{{\mathsf{c}},10}^2} -  \sqrt{f_{\frac{K}{2}}^2-f_{{\mathsf{c}},10}^2} \right) \Biggl).
\end{multline}
	The expression in \eqref{eq: chin} encapsulates the beam-squint effect of the wireless channel and DMA waveguide propagation since \eqref{eq: chin} describes the phase difference for the DMA beamforming weights between a specific subcarrier frequency and the center frequency of operation. We can then define the frequency-selective beamforming gain per subcarrier frequency including the approximate linear phase profile of the magnetic polarizability in the following lemma.
\begin{lemma}\label{lem: 4}

     The frequency-selective beamforming gain at each subcarrier frequency is
   \begin{equation}\label{eq: freq select}
        \Fsel(f_k)  = \left| \frac{1}{2}  \frac{1}{\sqrt{\Nt}} \frac{\sin\left( \frac{\Nt}{2}\chi_{\sfo}[k]\right)}{\sin \left( \frac{1}{2}\chi_{\sfo}[k] \right)}\right|^2.
\end{equation}
\end{lemma}

\begin{proof}
	Check Appendix~\ref{proof: lem4} for the proof.
	\end{proof}

It is important to note that since the DMA elements are assumed to have a linear phase profile from the polarizability approximation, the frequency-selectivity of the DMA elements has no impact on the overall frequency-selective beamforming gain. Under this assumption, the beam-squint of the wireless channel and DMA waveguide are the primary frequency-selective effects that degrade wideband performance for a DMA. This is valid for subcarrier frequencies near the resonant frequency tuning of the DMA elements and generally becomes less accurate as the signal bandwidth increases.

\subsection{Impact of tuning bandwidth}\label{sec: Impact of tuning bandwidth }

We now discuss the impact of the tuning bandwidth $\Rtb$ on the overall DMA performance. DMA weights do not achieve the full $[0,-\pi]$ phase range possible with the Lorentzian-constrained weight formulation in \eqref{eq: lor con} unless there is an infinite tuning bandwidth, which is not practically possible. The feasible weights are a function of both the varactor capacitance due to tuning bandwidth constraints and the damping factor. We use the term {\emph{weight fill}} to describe the ratio of Lorentzian-constrained weights available compared to the full set in \eqref{eq: lor con}. In the previous section, we made the assumption that the DMA elements had access to the full Lorentzian-constrained weights at the center frequency. Since the weight fill depends on the damping factor and tuning bandwidth, we develop an expression to account for the additional weight constraint of the weight fill on the final beamforming gain.

To derive a weight fill expression that describes how a weight fill impacts the DMA beamforming gain, we first introduce the weight fill ratio. The weight fill ratio $\Wratio(\Gamma,\Rtb)$ is described by the difference in phase of the lowest and highest resonant frequencies within the tuning bandwidth as
\begin{equation}\label{eq: weight ratio}
    \Wratio(\Gamma,\Rtb) = \frac{|\angle \tam(f_{\mathsf{t}}, f_{\sfr,{\mathsf{max}}}) - \angle \tam(f_{\mathsf{t}}, f_{\sfr,{\mathsf{min}}})|}{\pi},
\end{equation}
where it is normalized by $\pi$ to ensure the weight fill ratio remains between $[0,1]$. We also define $\xi = \pi \Wratio(\Gamma,\Rtb)$ to represent the angular weight fill of the DMA, with a range from $[0,\pi]$.
Next, we use the weight fill ratio to derive an approximation for the normalized beamforming gain from an incomplete Lorentzian-constrained curve as follows. 
\begin{lemma}\label{lem: 5}
    For a given weight fill ratio, the normalized beamforming gain is approximately
    \begin{equation}\label{eq: weight fill}
    \Wfill(\xi)  = \left| \frac{1}{2\pi} \left( 2\sin (\xi) + 2\xi \right) \right|^2.
\end{equation}
\end{lemma}
\begin{proof}
	Check Appendix~\ref{proof: lem5} for the proof.
\end{proof}

Based on the weight fill expression in \eqref{eq: weight fill}, we find that small weight fill values lead to a large amount of beamforming gain loss since there is a minimal phase range to use for beamforming. Then, as the weight fill increases, the normalized beamforming gain approaches $1$ as the phase range of the DMA beamformer approaches the full $[0,\pi]$ range. Therefore, a sufficient weight fill ratio, achieved through the DMA design, is necessary to ensure that the DMA can steer a beam pattern in a desired direction with minimal loss.

\subsection{Impact of waveguide attenuation}\label{sec: atten}

Lastly, we derive an expression to account for the impact of DMA waveguide attenuation due to power leakage on the MISO beamforming gain. As power is gradually radiated out through the DMA elements, the waveguide attenuation will alter the amplitude of each subsequent DMA element weight as an effective amplitude taper. This can potentially lead to additional losses. Let $\Prad$ be the radiated power from the DMA across the center subcarrier bandwidth and $\Lambda = \frac{\Prad}{\Pin}$ be the fractional radiated power. We derive a factor $\Aleak(\Lambda) \in [0,1]$ that describes the beamforming gain loss due to waveguide attenuation.

To define the penalty $\Aleak(\Lambda)$ independent of the DMA beamforming capability, we assume a frequency-selective matched-filter precoding $\bff_{\mathsf{mf}}[k]$ such that  $h_\nth[k]f_{\mathsf{mf},\nth}[k] = 1 \forall k,\nth$.
We make this assumption because we have separately accounted for the non-ideal effects of the frequency-selectivity, weight amplitude distribution, and tuning bandwidth in the previous sections. 
Therefore, with this assumption of ideal matched filtering and isolating the waveguide power leakage vector $\bh_{\mathsf{att}}[k]$, we define the penalty $\Aleak(\Lambda)$ as 
\begin{multline}\label{eq: bf gain atten}
\Aleak(\Lambda)=\frac{  | \heff ^{\mathsf{T}}[k](\mathbf{f}_{{\mathsf{mf}}}[k] \odot \bh_{\mathsf{att}}[k]) |^2}{\|\heff ^{\mathsf{T}}[k] \|^2  \| (\mathbf{f}_{{\mathsf{mf}}}[k] \odot \bh_{\mathsf{att}}[k]) \|^2} \\ = \frac{1}{\Nt} \frac{|  \mathbf{1}^{\sfT} \bh_{\mathsf{att}}[k]|^2}{\| \bh_{\mathsf{att}}[k] \|^2}.
\end{multline}
By above definition and from Cauchy-Schwarz inequality, we see that the maximum value of the penalty term is one. The final expression in \eqref{eq: bf gain atten} now depends only on the frequency-selective waveguide attenuation, and takes the form of a geometric series. Therefore, we simplify the  attenuation penalty term $\Aleak(\Lambda)$ in the following lemma.

\begin{lemma}\label{lem: 6}
Assuming ideal matched filtering and large $\Nt$, the DMA waveguide attenuation affects the ideal DMA beamforming gain by a penalty factor $ S(\Lambda)$ approximated as
\begin{equation}\label{eq: approx atten 1}
    \Aleak(\Lambda) \approx \frac{4}{\ln(1-\Lambda)} \tanh\left(\frac{\ln(1-\Lambda)}{4} \right).
\end{equation}
\end{lemma}
\begin{proof}
	Check Appendix~\ref{proof: lem6} for the proof.
\end{proof}

The derived expression in \eqref{eq: approx atten 1} gives the important insight that the DMA waveguide attenuation reduces the beamforming gain compared to a normal antenna array, depending on the total radiated power $\Prad$ and leakage constant $\atg$. We also justify the assumption for large $\Nt$ by noting that many DMA elements are required to provide a slow leakage of power and narrow beam patterns, meaning $\Nt$ must be large.

\section{DMA successive beamforming algorithm}\label{sec: successive}

We now describe two beamforming methods used to evaluate the wideband spectral efficiency of a DMA-based system. First, we define the DMA center-frequency beamforming algorithm as a baseline beamforming method to compare and validate the DMA beamforming gain approximation in \eqref{eq: approx bound}. Next, we define a novel, successive beamforming algorithm to improve upon the center-frequency beamforming method for larger signal bandwidths.

\subsection{DMA center-frequency beamforming algorithm}
The center-frequency beamforming algorithm selects the DMA resonant frequency tuning that best matches the channel phase at the center frequency. For the center frequency subcarrier $f_{\frac{K}{2}}$, the resonant frequency values for the center frequency beamformer are given by
\begin{equation}\label{eq: center beamformer}
f_{\sfr, \nth}^{\mathsf{cf}} = \argmin\limits_{f_{\sfr, \nth} \in [f_{\sfr,{\mathsf{min}}},f_{\sfr,{\mathsf{max}}}]} \left| - \frac{\sfj-e^{\sfj \angle h^{\mathsf{c}}_{\nth, \frac{K}{2}}}}{2} - \tam(f_{\frac{K}{2}},f_{\sfr, \nth}) \right|.
\end{equation}
Since the DMA approximation in \eqref{eq: approx bound} was derived based on a center frequency beamforming assumption, we will use the center frequency beamforming method in \eqref{eq: center beamformer} to evaluate the validity of the DMA approximation. 
\subsection{DMA successive beamforming algorithm}\label{subsec: successive alg}
We now develop a successive algorithm to configure the DMA elements and improve the DMA wideband spectral efficiency compared to the center-frequency beamformer. While the DMA approximation in \eqref{eq: approx bound} will help to give insights into wideband DMA element design, it does not provide insight into how the DMA element weights should be configured to improve the wideband spectral efficiency performance. The objective of the successive beamformer is to increase spectral efficiency by configuring each DMA element sequentially such that spectral efficiency is maximized for every subsequent element. This differs from the center-frequency algorithm which selects the DMA configuration solely based on the phase total channel.

To define the beamforming algorithm, let $f_{\sfr, \nth}^\sfS$ be the resonant frequency selected for the $\nth$th element of the successive beamformer, and let $U(f_k,f_{\sfr, \nth}) = \tam(f_k,f_{\sfr, \nth}) 
h_{\mathsf{att},\nth}[k] h_\nth[k]$ represent the per-element beamforming gain for the DMA with the wireless and waveguide channel. We define the selection process for the $\nth$th DMA element through the successive beamforming algorithm as
\begin{multline}\label{eq: successive}
f_{\sfr, \nth}^{\mathsf{S}} = \argmax\limits_{f_{\sfr, \nth} \in [f_{\sfr,{\mathsf{min}}},f_{\sfr,{\mathsf{max}}}]} \frac{1}{K}  \sum\limits_{k=1}^{K} \log_2 \big( 1+ \rho_k \big| U(f_k,f_{\sfr, \nth}) \\ + \sum\limits_{m=0}^{\nth-1}  U(f_k,f^\sfS_{\sfr, m}) \big|^2 \big).
\end{multline}
The successive algorithm in \eqref{eq: successive} first selects the resonant frequency that maximizes the wideband spectral efficiency for the DMA element closest to the waveguide feed. Then, with the chosen, fixed resonant frequency value of the first element, the successive algorithm then selects the resonant frequency for the second DMA element that maximizes the total spectral efficiency. We repeat this process until all resonant frequency values are selected to create the DMA successive beamforming vector.

The DMA successive algorithm has two primary advantages over the center-frequency beamforming algorithm. First, the successive algorithm maximizes the spectral efficiency over all frequencies in the bandwidth, rather than focusing solely on the center frequency. Second, the resonant frequency  selection process takes into account the previously-selected resonant frequencies to improve spectral efficiency, while the center-frequency algorithm simply attempts to match the DMA weight to the phase of the channel. We compare these two approaches in Section \ref{sec: results}.

\blue{We now compare the complexity of the proposed successive beamforming algorithm with the center-frequency beamforming algorithm. For both the center-frequency beamforming algorithm in \eqref{eq: center beamformer} and the successive beamforming algorithm in \eqref{eq: successive}, resonant frequency values are calculated from a continuous set $f_{\sfr} \in [f_{\sfr,{\mathsf{min}}},f_{\sfr,{\mathsf{max}}}]$. When implementing these algorithms computationally to calculate the resonant frequencies, however, the possible continuous resonant frequency values must be transformed into a discrete set with resolution $\Rres$, such that $f_{\sfr} \in \{f_{\sfr,{\mathsf{min}}},f_{\sfr,{\mathsf{min}}} + 1\frac{f_{\sfr,{\mathsf{max}}} - f_{\sfr,{\mathsf{min}}}}{\Rres-1}, f_{\sfr,{\mathsf{min}}} + 2\frac{f_{\sfr,{\mathsf{max}}} - f_{\sfr,{\mathsf{min}}}}{\Rres-1}, \ldots, f_{\sfr,{\mathsf{min}}} + (\Rres-2)\frac{f_{\sfr,{\mathsf{max}}} - f_{\sfr,{\mathsf{min}}}}{\Rres-1}, f_{\sfr,{\mathsf{max}}}\}$. To calculate the resonant frequency for each $\nth$th element with $\Nt$ total elements, the minimization for \eqref{eq: center beamformer} and maximization for \eqref{eq: successive} then occurs over the discrete resonant frequency set with $\Rres$ values. Therefore, the complexity of both the center-frequency beamforming algorithm and the successive beamforming algorithm scale as $O(\Nt \Rres)$ when performing the computations. This means that the proposed algorithm does not impose a greater complexity scaling than the baseline center-frequency beamforming algorithm. }

\blue{Lastly, we compare the proposed algorithm with prior work in Table \ref{table: prior work comparison} based on the various DMA modeling parameters that are incorporated into the algorithm. We find that our proposed successive beamforming algorithm is the only work that accounts for the frequency-selectivity of the DMA elements, tuning bandwidth, and waveguide power leakage all together, where \cite{smith2017analysis,boyarsky2021electronically} and \cite{WangEtAlDynamicMetasurfaceAntennasMIMOOFDM2021,HuangEtAlStructuredOFDMModulationXLMIMO2025} account for different combinations of these three parameters. Moreover, we have designed the successive beamforming algorithm to improve the spectral efficiency in wideband scenarios, which differs significantly from all prior work that has primarily focused on narrowband scenarios. We evaluate the performance of the proposed successive beamforming algorithm for large bandwidths in the next section. }

% \begin{center}
% \renewcommand{\arraystretch}{2}
% \begin{tabular}{||c c c c||} 
%  \hline
%  Parameter & Col2 & \cite{WangEtAlDynamicMetasurfaceAntennasMIMOOFDM2021} & This work \\ [0.5ex] 
%  \hline\hline
%  Narrowband & & \checkmark & \checkmark \\
%  \hline
%  Wideband & & & \checkmark \\
%  \hline
%   Frequency-selectivity &  & \checkmark & \checkmark \\
%  \hline
%  Tuning bandwidth & &  & \checkmark \\ 
%  \hline
%  Waveguide power leakage &  &  & \checkmark \\
%  \hline
% \end{tabular}
% \end{center}

\begin{table}
\centering
\begin{center}
\caption{\blue{Summary of DMA model parameters incorporated into the beamforming algorithm design}}
\label{table: prior work comparison}
\begin{tabular}{|| >{\centering\arraybackslash}m{1.8cm} | >{\centering\arraybackslash}m{1.4cm} | >{\centering\arraybackslash}m{.9cm} | >{\centering\arraybackslash}m{1.2cm} | >{\centering\arraybackslash}m{0.8cm} || } 
\hline
 Parameter & \cite{deng2022reconfigurable,hwang2020binary,ShlezingerEtAlDynamicMetasurfaceAntennasUplink2019,WangJointTransmitterReceiverDesign2022,YouEtAlEnergyEfficiencyMaximizationMassive2022,ZhangEtAlBeamFocusingNearFieldMultiuser2022,KimaryoLeeDownlinkBeamformingDynamicMetasurface2023,AzarbahramEtAlEnergyBeamformingRFWireless2023,HuangEtAlJointMicrostripSelectionBeamforming2023} & \cite{smith2017analysis,boyarsky2021electronically} & \cite{WangEtAlDynamicMetasurfaceAntennasMIMOOFDM2021,HuangEtAlStructuredOFDMModulationXLMIMO2025} & This work \\ 
 \hline
 \hline
 Narrowband & \checkmark & \checkmark & \checkmark  & \checkmark \\ 
 \hline
 Wideband & & & & \checkmark \\ 
 \hline
 Frequency-selective &  & & \checkmark & \checkmark \\ 
 \hline
 Tuning bandwidth &  & \checkmark & & \checkmark \\ 
 \hline
 Waveguide power leakage & & \checkmark & & \checkmark \\ 
 \hline
\end{tabular}
\end{center}
\end{table}

% It should be noted that in practice, the resolution $\Rres$ that determines the number of possible discrete resonant frequency values is dictated by the digital-to-analog converter resolution, which controls the bias voltage to each varactor. We assume this to be 

\section{Data rate and spectral efficiency results}\label{sec: results}

We now examine the data rate and spectral efficiency results for the center-frequency DMA beamformer, the approximate model, and the DMA successive beamformer. For specific DMA design parameters, each value in the approximate model is calculated to approximate the total DMA beamforming gain for spectral efficiency and data rate results, while the center-frequency DMA beamformer and successive beamformer configures the DMA elements based on \eqref{eq: center beamformer} and \eqref{eq: successive}, respectively.

\subsection{Simulation setup and design considerations}\label{subsec: Design considerations}

In all cases, the spectral efficiency $C(\{ \rho_k, \heff[k]\}_{k=1}^K) $ is calculated using \eqref{eq: spec eff} with the configured DMA elements. The channel in every case is an LOS channel at a particular target direction $\phit$. Data rates $D$ are then calculated for the signal bandwidth $B$ as
\begin{align}
D = B C(\{ \rho_k, \heff[k]\}_{k=1}^K).
\end{align}
Design parameters for the DMA, such as tuning bandwidth, damping factor, and number of elements, will be varied and analyzed for their impact on the resulting spectral efficiency and data rates.

Next, we discuss the frequency-selective DMA and waveguide design parameters for simulations. As discussed in Lemma 6, leaky-wave antennas are typically designed such that a predetermined fraction $\Lambda$ of the input power is radiated. The desired leakage constant $\bar{\alpha}_\sfg$ at the center frequency for a typical leaky-wave antenna is \cite{oliner2007leaky}
\begin{equation}\label{eq: desired atten}
    \bar{\alpha}_\sfg(\ft) = \frac{\ln(1-\Lambda)}{2 \dx (\Nt-1)}.
\end{equation}
As an example, designing the DMA leakage constant $\atg$ for $\Lambda = 0.9$ in \eqref{eq: desired atten} ensures that $90\%$ of the input power is radiated. Calculating the attenuation due to the radiated power of DMA elements, however, is a more difficult task since each DMA element has its own unique element configuration and radiated power. We follow a procedure outlined in \cite{smith2017analysis} to determine the DMA waveguide leakage as a function of frequency based on an averaged value of the different element configurations. 

The attenuation model in \cite{smith2017analysis} implements ideal waveguide assumptions such that there is no power loss due to waveguide imperfections, and only power radiated through the DMA elements contributes to the electromagnetic field decay. A more complete model for the attenuation of the waveguide fields involves calculating the radiated electromagnetic fields from each DMA element and subtracting the radiated power from the waveguide fields. This method depends on the specific configuration of DMA element weights and complicates the analysis of the impact of waveguide field attenuation on the wideband performance significantly. Therefore, we use the averaged attenuation to simplify the wideband analysis and allow for the derivation of the leakage term $\Aleak$. Moreover, we assume that we can alter the DMA coupling factor $\Fcoup$ to achieve the desired attenuation constant $\atg$ for the DMA.

The normalization of the DMA beamformer is based on the radiated power per subcarrier bandwidth. For a typical fully-digital MISO system, a beamformer $\mathbf{f}[k]$ is subject to a power constraint for a normalization constant $M$ such that $\sum\limits_{k=1}^K ||\sqrt{M} \mathbf{f}[k]  ||^2 = K \Pin$. With this normalization comes the assumption that 100\% of the input power is radiated out through the antenna array. The radiated power for the DMA, however, is given by \cite{oliner2007leaky}
\begin{equation}\label{eq: prad}
    \Prad(f_k) = \Pin\left(1 - e^{-2\atg(f_k)\dx (\Nt-1)}\right)
\end{equation}
and describes the DMA radiated power across each subcarrier bandwidth, which does not guarantee that 100\% of the input power is radiated. We reformulate the power constraint for all subcarriers based on \eqref{eq: prad} and configure the power normalization constant $M_k$ for $\fdma[k]$ such that
\begin{equation}\label{eq: precoder norm}
||\sqrt{M_k} \fdma[k] \odot \mathbf{h}_\mathsf{att}[k]  ||^2 = \Pin \left(1 - e^{-2\atg(f_k)\dx (\Nt-1)}\right) , \; \forall k
\end{equation}
The power constraint in \eqref{eq: precoder norm} ensures that the frequency-selectivity of the radiated power in \eqref{eq: prad} for a specific DMA design is physically consistent with the DMA beamforming weights in $\fdma$ through the normalization constant $M_k$.

Since we have an averaged attenuation model, the normalization procedure comes with the assumption that regardless of the DMA element configuration, a certain radiated power value is met. Because of the waveguide phase advance, the DMA element configuration to steer a beam pattern will have a somewhat uniform distribution of the antenna weights around the Lorentzian-constraint, as discussed in Lemma 5. All DMA element configurations will have a very similar radiated power, allowing us to justify the radiated power assumption.

\subsection{Validation for the DMA beamforming gain approximation}\label{sec: validation}

\begin{figure} 
    \centering
  \subfloat[\label{fig: spec eff bw}]{%
       \includegraphics[width=\figsizeii\linewidth]{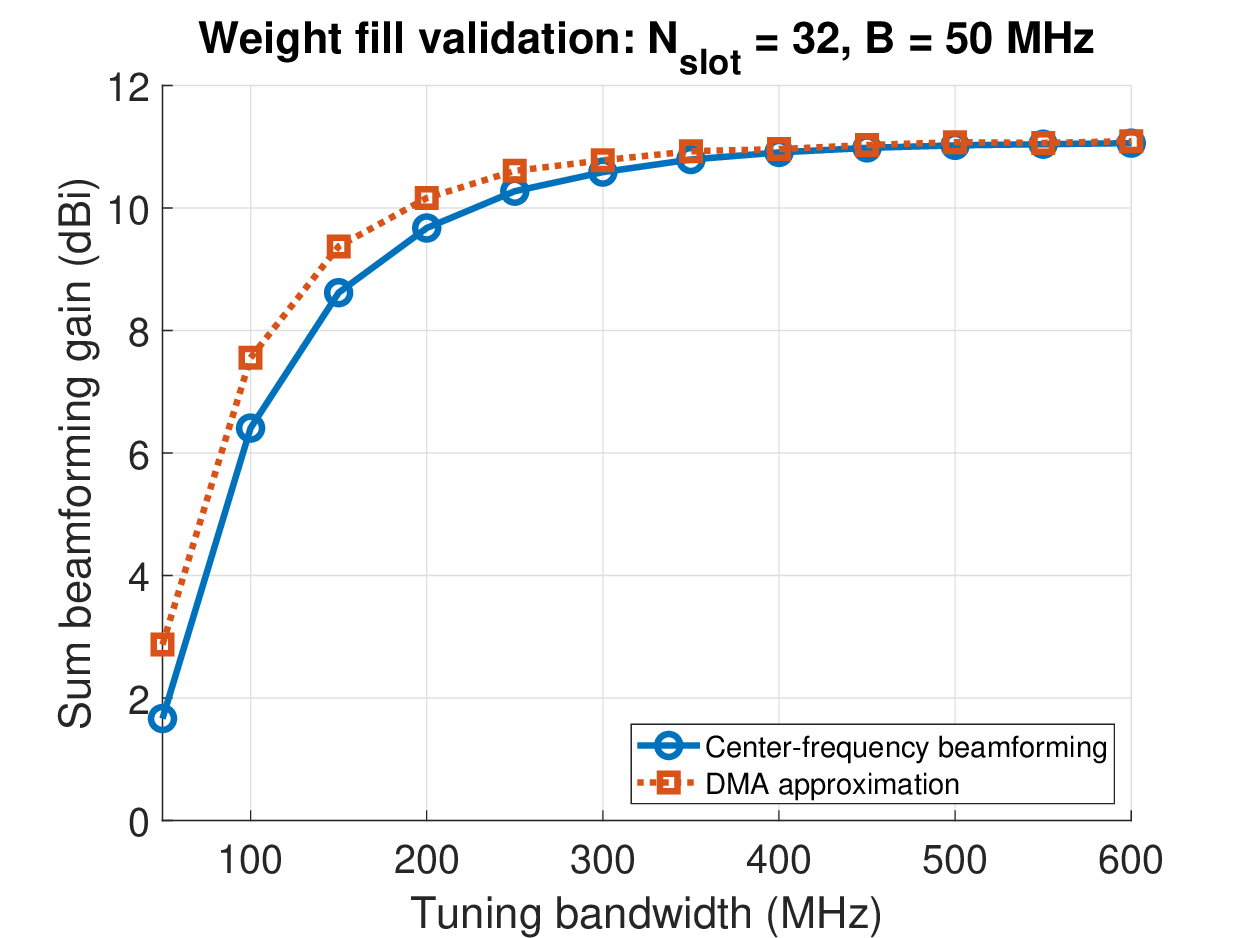}}
    \hfill
  \subfloat[\label{fig: data rate bw 2}]{%
        \includegraphics[width=\figsizeii\linewidth]{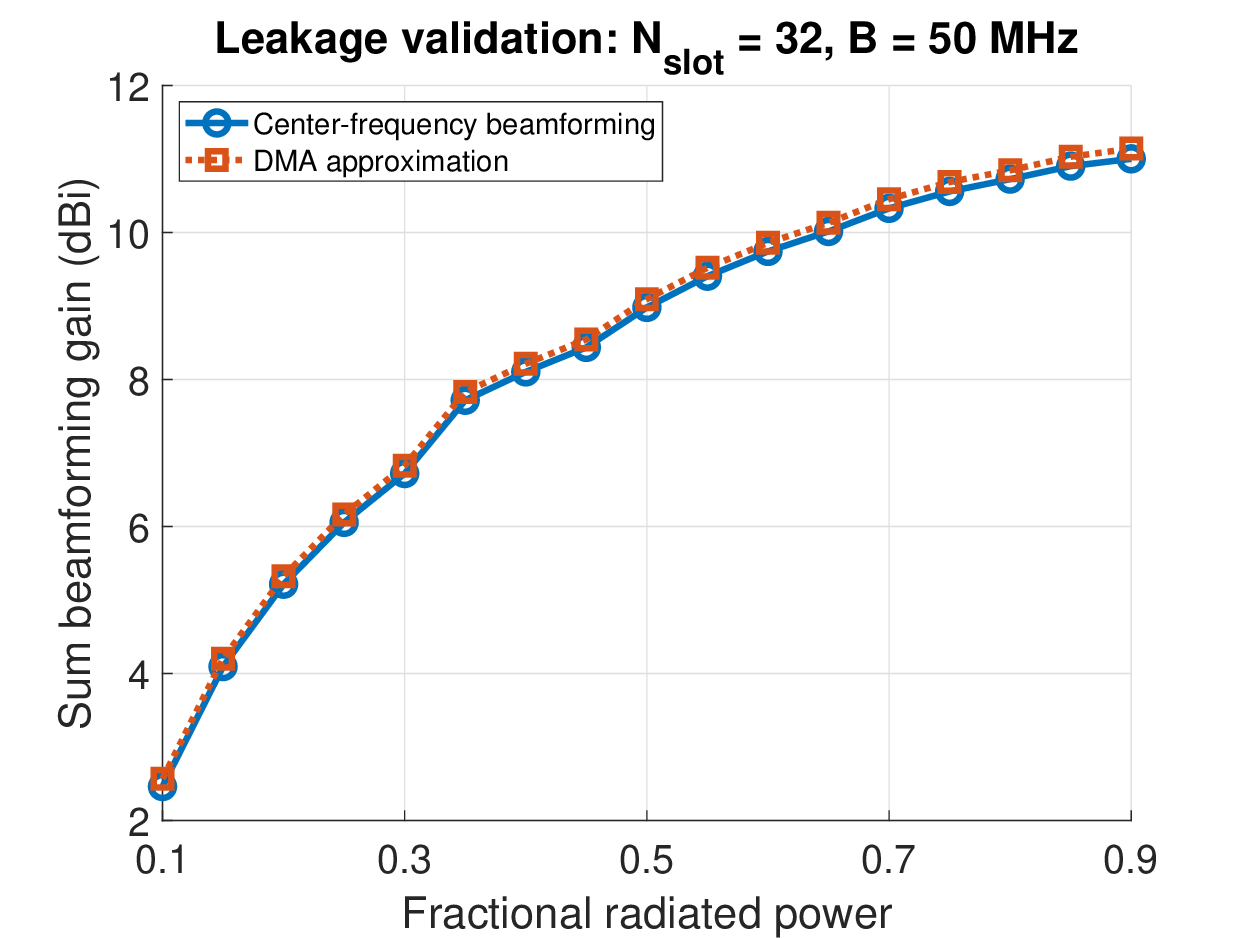}}
   \hfill
  \subfloat[\label{fig: data rate bw 3}]{%
        \includegraphics[width=\figsizeii\linewidth]{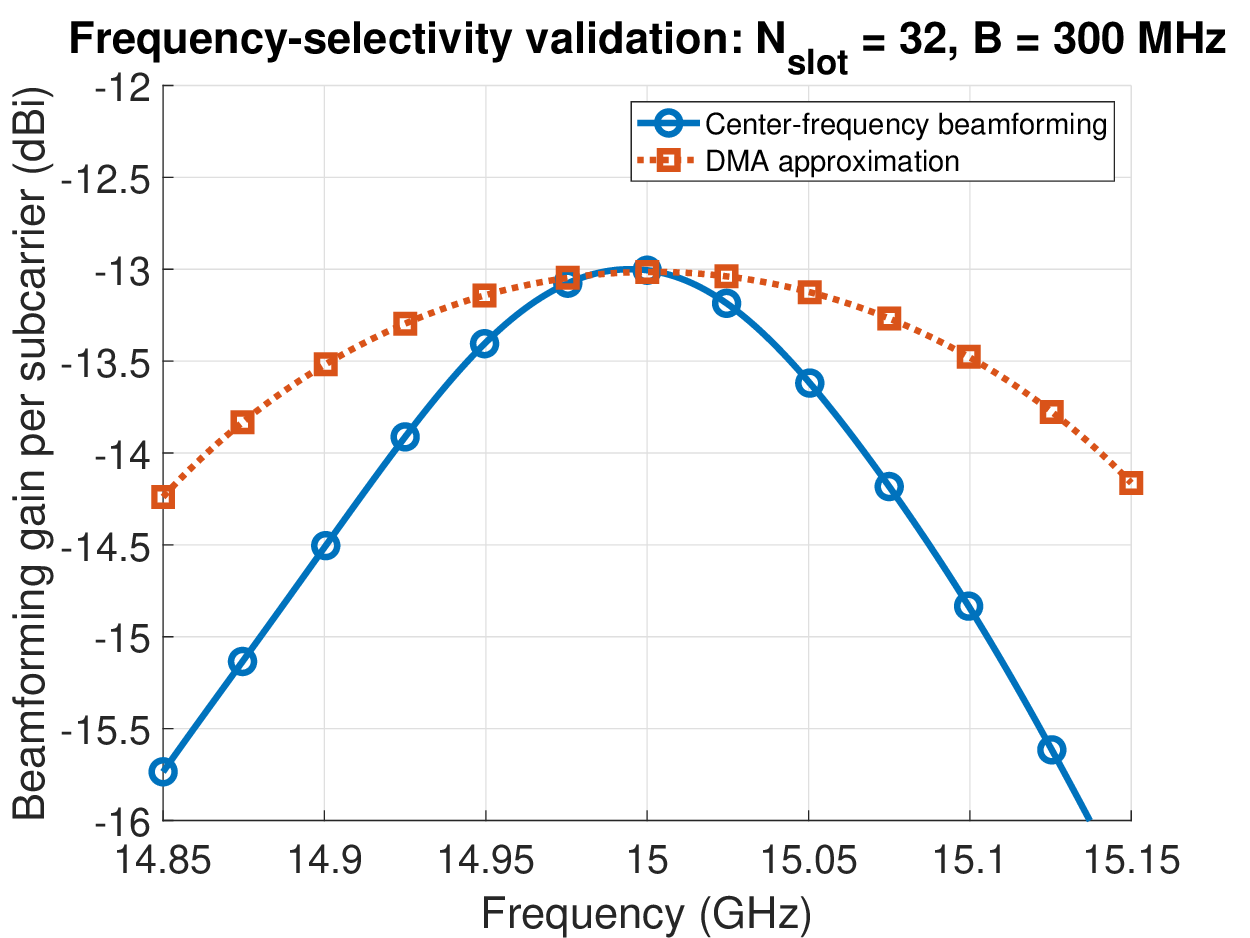}}

  \caption{We plot the validation for the (a) weight fill term, (b) leakage term, and (c) frequency-selectivity term using the DMA approximation against the center-frequency beamformer. We find that the weight fill and leakage terms provide an accurate model to approximate the wideband DMA performance, while the frequency-selectivity term is accurate at small signal bandwidths and becomes less accurate as the signal bandwidth increases.}
  \label{fig: validations} 
\end{figure}

Next, we validate the accuracy of the DMA beamforming gain approximation expression in \eqref{eq: approx bound}. Unless otherwise specified, we consider a DMA design that operates at $\ft = 15$ GHz, has $\Nt=32$ elements with an element spacing of $\dx = \lambda/4$, and has a quality factor of $Q = 100$. We will isolate the three approximation terms $\Fsel, \Wfill,\Aleak$ and compare the beamforming gain with the center-frequency beamformer to determine validity of each term separately. 

To validate the approximation in \eqref{eq: approx bound}, we denote the per-subcarrier beamforming gain for the center-frequency beamformer as $G_\mathsf{cf}[k] = |\heff^{\mathsf{T}}[k] \left(\fdma[k] \odot \mathbf{h}_{\mathsf{att}}[k] \right)|^2 $ and for the approximation as $G_\mathsf{approx}[k] = \Fsel(f_k)\Wfill(\xi) \Aleak(\Lambda)$. We then define the sum beamforming gain for the center-frequency beamformer as 
\begin{equation}
    G_\mathsf{cf,sum} = \sum\limits_{k=1}^{K} M_k|\heff^{\mathsf{T}}[k] \left(\fdma[k] \odot \mathbf{h}_{\mathsf{att}}[k] \right)|^2
\end{equation}
and for the approximation as 
\begin{equation}
    G_\mathsf{approx,sum} = \sum\limits_{k=1}^{K} M_k \Fsel(f_k)\Wfill(\xi) \Aleak(\Lambda).
\end{equation}
We validate the DMA approximation using the per-subcarrier and sum beamforming gain as follows.

First, we compare $G_\mathsf{cf,sum}$ and $G_\mathsf{approx,sum}$ for a small bandwidth to isolate the weight fill term $\Wfill$. For a small bandwidth and fixed radiated power, the frequency-selective term and and leakage term will not have a significant impact on the resulting sum beamforming gain, while the weight fill term $\Wfill$ changes significantly as a function of the tuning bandwidth $\Rtb$. We show the sum beamforming gain for the center-frequency beamforming and approximation in Fig. \ref{fig: validations}(a). As the tuning bandwidth increases, we find that the approximation closely matches the sum beamforming gain of the center-frequency beamformer, meaning that the weight fill term provides an accurate analysis of how the DMA beamforming gain changes as a function of the tuning bandwidth. As expected, we also find that as the tuning bandwidth grows, the sum beamforming gain increases due to the additional available beamforming weights.

Next, for a small signal bandwidth and large tuning bandwidth, we show the validation for the leakage term $\Aleak$ through the sum beamforming gain as a function of the fractional radiated power $\Lambda$ in Fig. \ref{fig: validations}(b). Here, we see that the approximation very closely matches the center-frequency beamformer and that the sum beamforming gain increases as we increase the power radiated from the DMA. 

Lastly, for a large tuning bandwidth and fixed radiated power, Fig. \ref{fig: validations}(c) shows the per subcarrier beamforming gain as a function of the subcarrier frequency. Using a very large tuning bandwidth, we isolate the frequency-selective term $\Fsel(f_k)$ to compare with the center-frequency beamformer. In Fig. \ref{fig: validations}(c) at frequencies near the center frequency $\ft = 15$ GHz, the per subcarrier beamforming gain matches very closely to the center-frequency beamforming method and becomes less accurate at subcarrier frequencies further from the center frequency. This is due to the linear approximation for the DMA frequency-selectivity in \eqref{eq: ang pol angle approx}, which becomes less accurate as the subcarrier frequency deviates further from the center frequency. Overall, the leakage term $\Aleak$ and the weight fill term $\Wfill$ provide an accurate model for the DMA beamforming gain as the tuning bandwidth and fractional radiated power vary, while the frequency-selectivity term $\Fsel$ is accurate at smaller signal bandwidths and becomes less accurate at larger signal bandwidths.

\subsection{Simulation results for the approximation and successive DMA beamformer}\label{sec: successive results}

\begin{figure} 
    \centering
  \subfloat[\label{fig: spec eff bw 1}]{%
       \includegraphics[width=\figsizeii\linewidth]{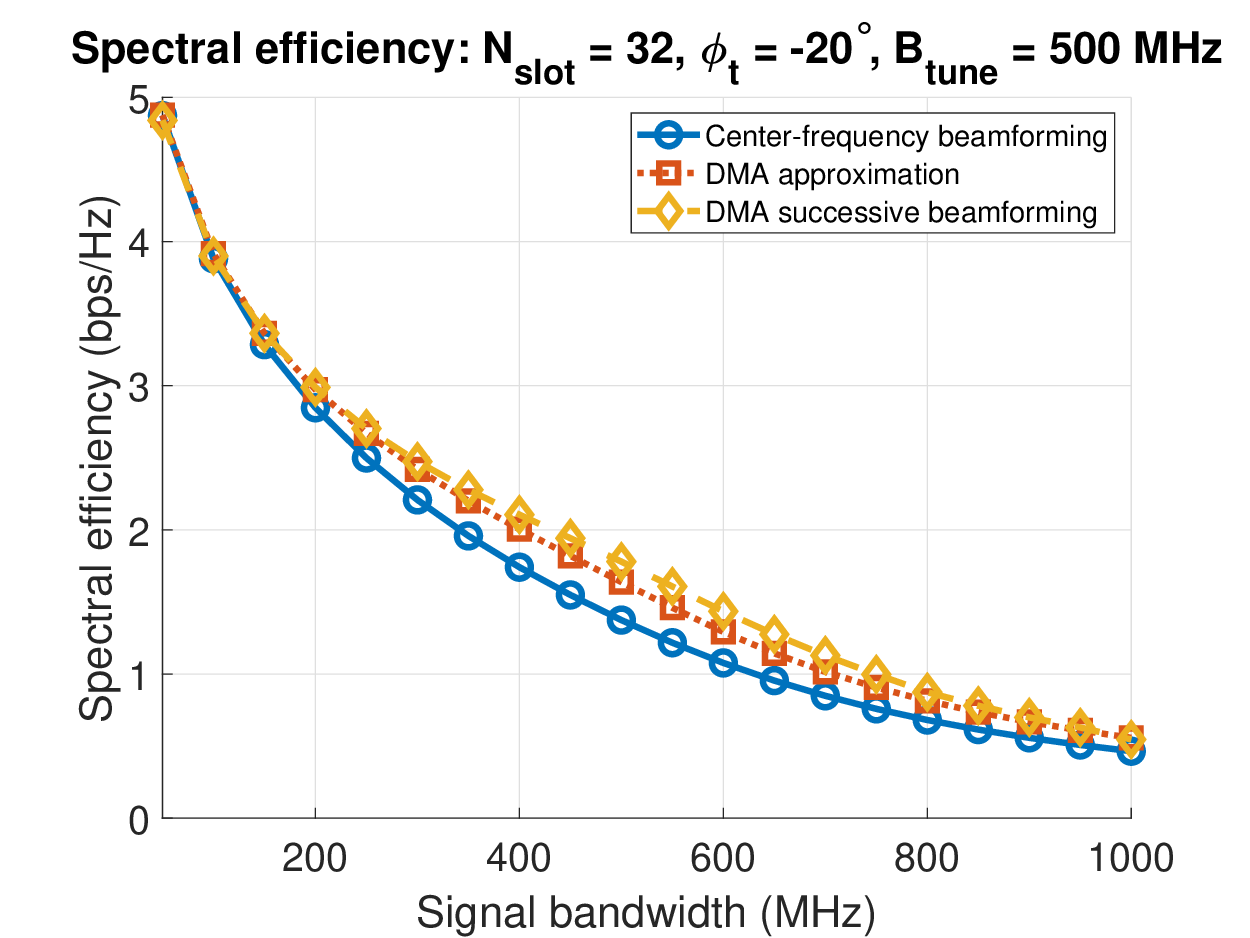}}
    \hfill
  \subfloat[\label{fig: data rate bw 1}]{%
        \includegraphics[width=\figsizeii\linewidth]{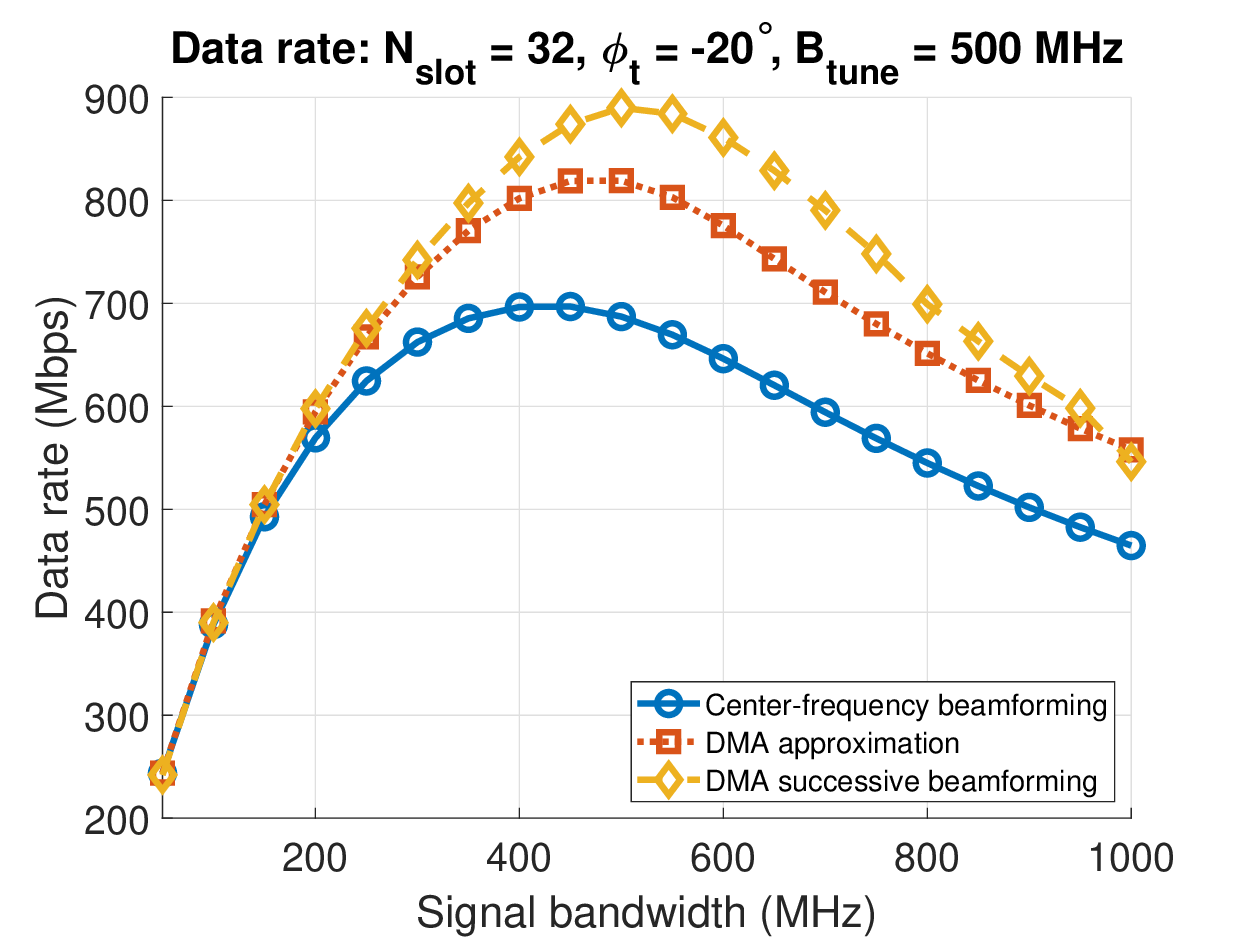}}

  \caption{We plot  spectral efficiency in (a) and  data rate in (b) as a function of the signal bandwidth. We find that spectral efficiency decreases with signal bandwidth due to the noise power term and increased frequency-selectivity of the DMA elements. We also see that the successive algorithm provides better data rates and spectral efficiency than the center-frequency beamformer.}
  \label{fig: wideband data rates} 
\end{figure}

We now investigate the performance of the DMA successive algorithm compared to the baseline center-frequency beamforming performance. Fig. \ref{fig: wideband data rates}(a) shows the DMA wideband spectral efficiency as a function of signal bandwidth for a large tuning bandwidth. As expected, we find that the spectral efficiency decreases with signal bandwidth due to the increased noise power and higher DMA frequency selectivity. Given the large tuning bandwidth, the DMA successive algorithm is able to increase the spectral efficiency compared to the baseline center-frequency beamformer since the successive algorithm better configures the DMA elements across a wide signal bandwidth.

\begin{figure}
    \centering
    \includegraphics[width=\figsizeii\linewidth]{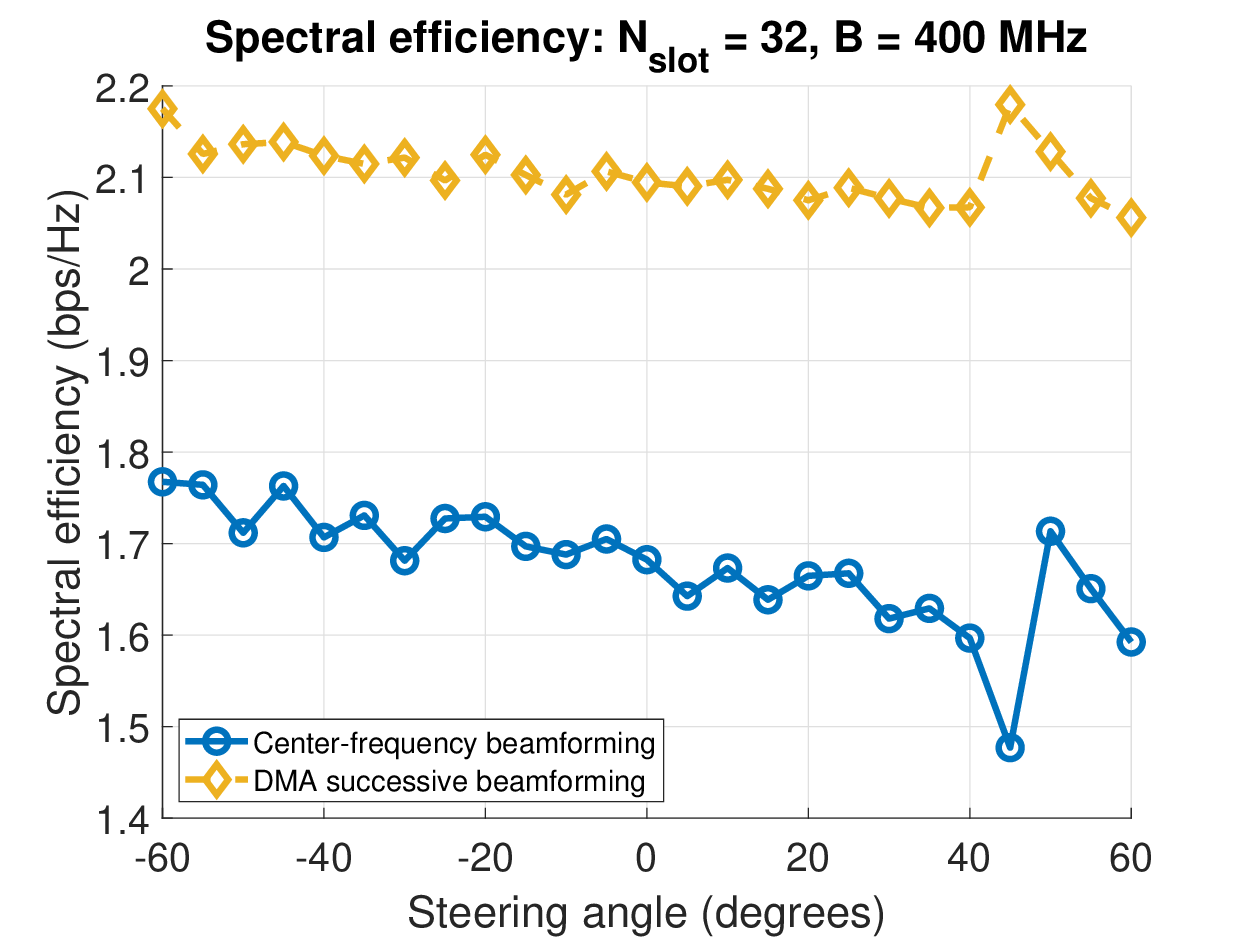}
    \caption{Spectral efficiency as a function of the steering angle $\phit$. As the steering angle varies, the spectral efficiency does not change significantly for both the center-frequency beamformer and the successive beamformer.}
    \label{fig: spec eff angles}
\end{figure}

Fig. \ref{fig: wideband data rates}(b) also shows the data rates for the DMA wideband scenario. We see that as the signal bandwidth increases, the data rate continues to increase until the DMA elements exhibit very large frequency-selectivity. From the attenuation analysis, we see that the tuning bandwidth provides an approximate limit for the maximum data rate, since the DMA cannot support frequencies beyond the tuning bandwidth. The approximation holds for smaller signal bandwidths, but begins to breakdown at larger signal bandwidths due to the magnetic polarizability approximation in \eqref{eq: ang pol angle approx}. We use the insights in the approximation to understand how the frequency-selectivity, tuning bandwidth, and attenuation impact the DMA wideband performance in Section \ref{subsec: impact of dma el spacing}. Lastly, Fig. \ref{fig: spec eff angles} shows the spectral efficiency as a function of the steering angle. We find that the spectral efficiency does not change significantly as the steering angle varies, meaning the we can extend the insights gained from Fig. \ref{fig: wideband data rates} at the angle $\phit = -20^\circ$ to all steering angles.

We now discuss the main benefit of the DMA successive algorithm to enhance the achievable data rates of the DMA wideband system. Fig. \ref{fig: max data rates} shows the maximum data rate achievable for a given tuning bandwidth as $D_{\mathsf{max}} = \max\limits_{B} D(B,\Rtb)$. For both the DMA center-frequency beamformer and the successive algorithm, we find that there is ultimately a data rate plateau as we increase the DMA tuning bandwidth. This is due to the high frequency-selectivity of the DMA system, where increasing the tuning bandwidth no longer yields greater wideband beamforming capabilities. We find that the successive algorithm, however, can significantly extend the maximum achievable data rate at higher signal bandwidths by configuring the frequency-selective DMA elements for the wideband signal. Therefore, the proposed successive algorithm can perform significantly better than the baseline algorithm for a large signal bandwidth and tuning bandwidth.

\begin{figure}
    \centering
    \includegraphics[width=\figsizeii\linewidth]{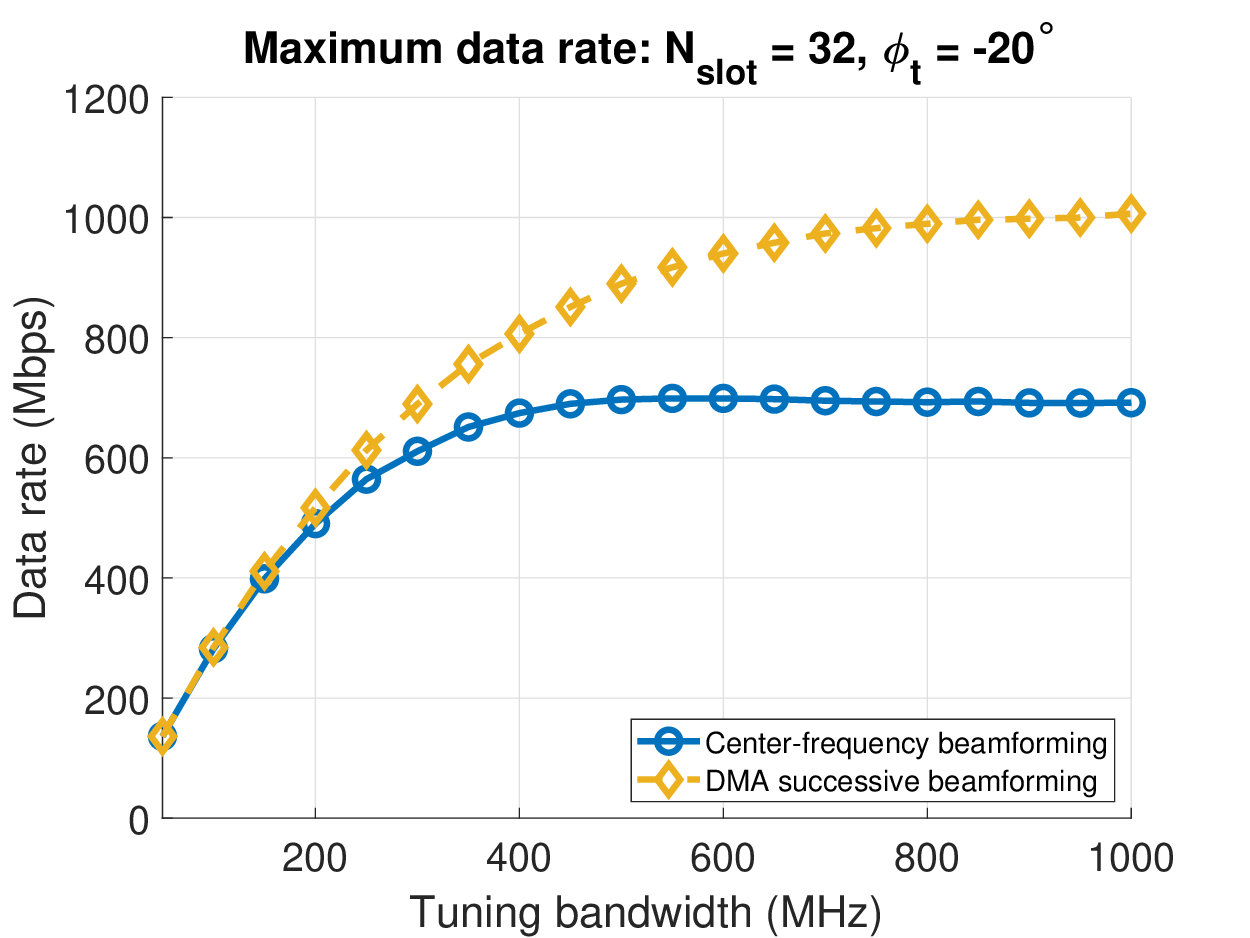}
    \caption{Maximum data rate results comparing the center-frequency beamforming and DMA successive algorithm. The DMA successive algorithm allows for larger data rates and signal bandwidths as it can better configure the DMA elements for a wideband frequency-selective scenario.}
    \label{fig: max data rates}
\end{figure}

\subsection{Impact of DMA element spacing and damping factor}\label{subsec: impact of dma el spacing}

\begin{figure}
    \centering
    \includegraphics[width=\figsizeii\linewidth]{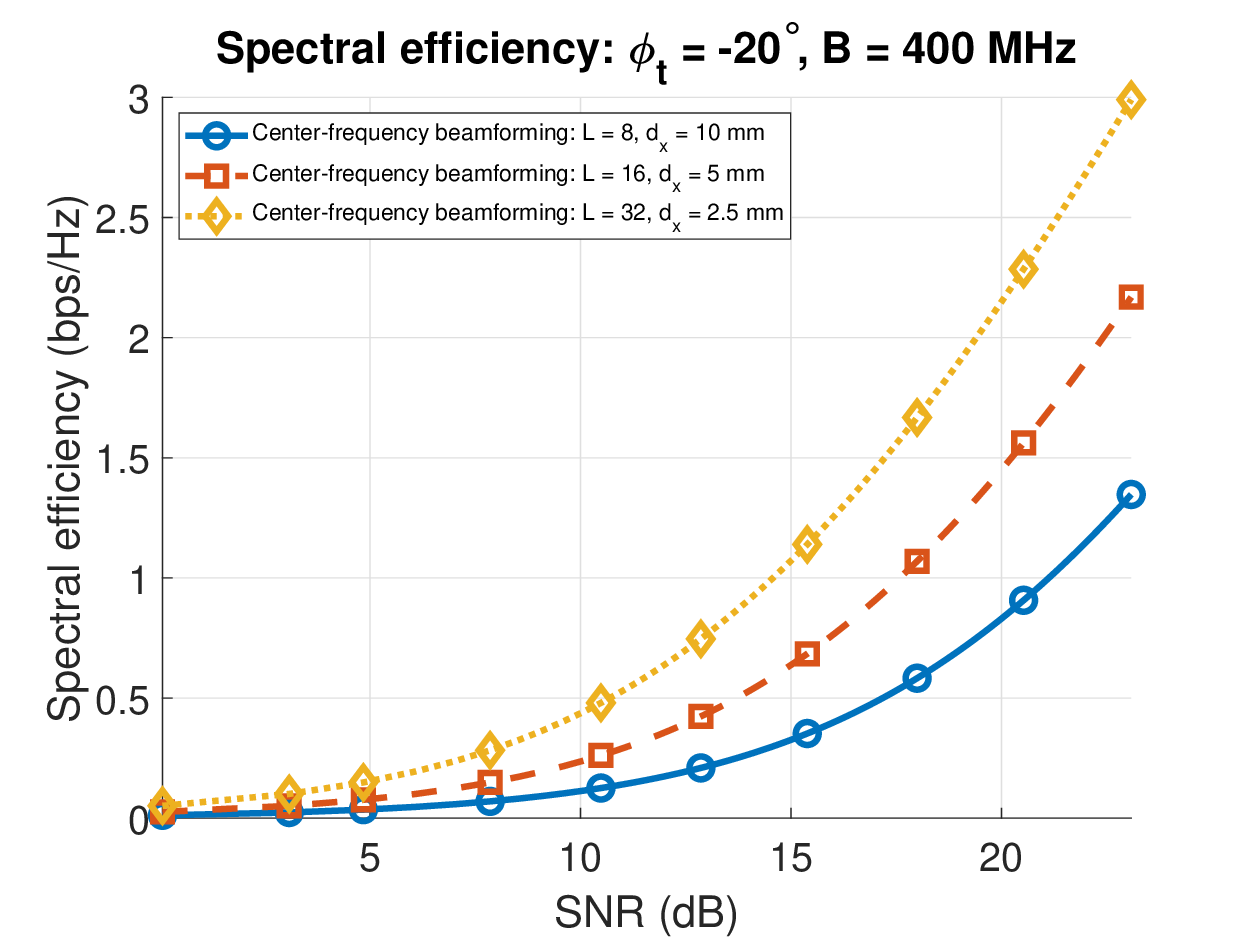}
    \caption{DMA element spacing analysis for spectral efficiency. As the approximation predicts, smaller DMA element spacing leads to lower frequency-selectivity and higher spectral efficiency. As DMAs are impacted more by the beam-squint effect than phased arrays, it is important to minimize DMA element spacing to mitigate beam-squint.}
    \label{fig: spacing}
\end{figure}

We now look at the impact of the DMA element spacing and the damping factor on the resulting DMA wideband spectral efficiency. Fig. \ref{fig: spacing} shows spectral efficiency results for the center-frequency beamformer as a function of the DMA element spacing. From the frequency-selectivity term in \eqref{eq: freq select}, we expect the frequency-selectivity to increase as the DMA element spacing increases. Lower element spacing will be less frequency-selective and result in improved spectral efficiency. Fig. \ref{fig: spacing} verifies this relationship, where spectral efficiency increases as $\dx$ decreases, while maintaining the same overall DMA length. This also leads to key insight on how the DMA frequency-selectivity compares to a phased array. The addition of the waveguide channel greatly increases the frequency-selectivity. Therefore, small dense element spacing is crucial for a DMA wideband system to minimize the impact of the frequency-selective waveguide and wireless channel.

Next, we examine the effects of the damping factor on the DMA performance. Fig. \ref{fig: damping fac} shows the data rate results for a high, medium and low damping factor. We define the damping factor in terms of the Q-factor at the center frequency as $\Gamma = \frac{2\pi \ft}{ Q}$. Larger damping factors leads to lower amounts of frequency-selectivity in the magnetic polarizability resonance curve but also a lower weight fill ratio. We show in Fig. \ref{fig: damping fac} that a large damping factor leads to the highest data rates for a large tuning bandwidth due to lower frequency-selectivity.

\begin{figure}
    \centering
    \includegraphics[width=\figsizeii\linewidth]{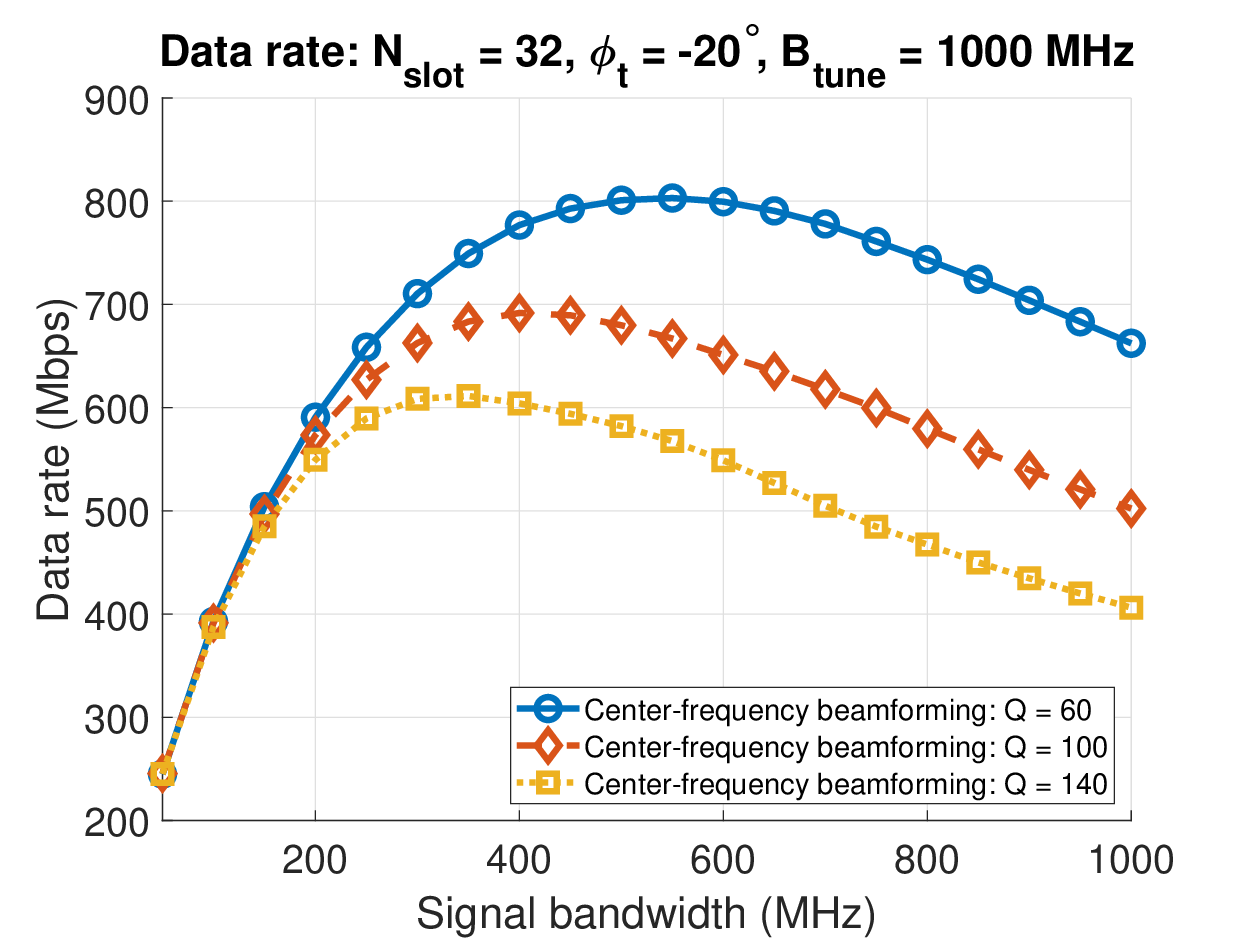}
    \caption{Data rates for the center-frequency beamformer and different damping factors. We show that larger damping factors lead to improved data rates due to the lower impact of frequency-selectivity.}
    \label{fig: damping fac}
\end{figure}

\blue{\subsection{Impact of multipath on DMA beamforming performance}\label{subsec: multipath}

Lastly, we determine the effectiveness of the proposed successive beamforming algorithm in a multipath environment and make comparisons to hybrid and analog beamforming architectures. Thus far, we have primarily considered an LOS scenario to derive the wideband beamforming gain approximation and extract important design insights for a wideband DMA system. Since practical channel environments often involve multiple propagation paths, we include multipath channel simulation results to ensure that our proposed algorithm remains effective in various channel environments. A more in-depth analysis on the impact of multipath environments for DMA wideband beamforming can be found in our related work \cite{j_icc}. Here, we extend the work in \cite{j_icc} to focus on the comparison of DMAs with typical hybrid beamforming architectures that employ passive phase shifters.

Next, we discuss the simulation parameters and setup. The multipath channel model is outlined in \cite{j_icc}, which includes $L_\mathsf{path}$ channel paths and frequency-selective filtering effects. We do not repeat the multipath channel model here for brevity. We also consider a uniform planar array (UPA) architecture at the transmitter instead of a ULA, where antenna elements are also placed in the $y$ direction with spacing $d_\mathsf{y}$ to form subarrays. For the hybrid beamforming architecture and analog-only architecture, let there be $N_\mathsf{ele}$ antenna elements per subarray with spacing $\dx=\lambda/2$ in the $x$ direction, and $N_\mathsf{rf}$ subarrays with spacing $d_\mathsf{y}=\lambda/2$ in the $y$ direction. Although a spacing of $\dx = \lambda/2$ is common in phased arrays due to antenna sizes, mutual coupling, and grating lobe mitigation, DMA elements are typically much smaller than traditional antenna elements and can allow for tight element spacing. This is due to a combination of the metasurface-inspired geometry as well as the slotted-waveguide implementation of DMAs. Because DMA elements can be placed more densely than typical antenna elements, let there be $2N_\mathsf{ele}$ antenna elements per subarray with spacing $\dx=\lambda/4$ in the $x$ direction, and $N_\mathsf{rf}$ subarrays with spacing $d_\mathsf{y}=\lambda/2$ in the $y$ direction. While the DMA will contain twice the number of antenna elements, the aperture area of both the DMA and hybrid architecture will be the same to maintain consistency between the two architectures. We calculate the beamforming vectors for both a partially-connected and fully-connected hybrid architecture based on \cite{ParkEtAlDynamicSubarraysHybridPrecoding2017,castellanos2023energy}. For the hybrid architectures, the number of subarrays $N_\mathsf{rf}$ dictates the number of RF chains. Lastly, it is crucial to account for the loss due to RF components when comparing DMAs with a hybrid architecture. We use the component loss and SNR model outlined in \cite{RibeiroEtAlEnergyEfficiencyMmWaveMassive2018} to determine the loss due to the passive phase shifters and power dividers for the DMA, analog-only, and hybrid architectures. We assume the loss due to the passive phase shifters is $8.8$ dB \cite{RibeiroEtAlEnergyEfficiencyMmWaveMassive2018}.

For the multipath channel environment with $N_\mathsf{ele} = 6$ elements, $N_\mathsf{rf} = 4$ RF chains, and $L_\mathsf{path} = 4$ channel paths, Fig. \ref{fig: spec eff multipath} shows the resulting spectral efficiency for the DMA and hybrid architectures through a Monte Carlo simulation. We vary the total transmit power $P_\mathsf{t}$, which modifies the input power $\Pin$ based on the component loss $L_\mathsf{loss}$ as $P_\mathsf{in} = L_\mathsf{loss}P_\mathsf{t}$. This is further discussed in \cite{carlson2023hierarchical}. Because the multipath environment introduces more frequency-selectivity into the wireless channel, we find here that the proposed DMA successive beamforming algorithm significantly outperforms the DMA center-frequency beamforming algorithm. This shows that our proposed algorithm remains very effective for channel environments beyond the LOS scenario. Moreover, we find that the DMA outperforms both the hybrid architectures and the analog-only architecture in terms of spectral efficiency. This is primarily due to the component loss in the passive phase shifters, which is set to $8.8$ dB \cite{carlson2023hierarchical} and contributes to a significant loss in the SNR compared with DMAs that employ no phase shifters. Overall, the proposed successive beamforming algorithm provides a large performance improvement in wideband wireless systems compared to the baseline DMA beamforming method and hybrid architectures that use passive, lossy phase shifters.

\begin{figure}
    \centering
    \includegraphics[width=\figsizeii\linewidth]{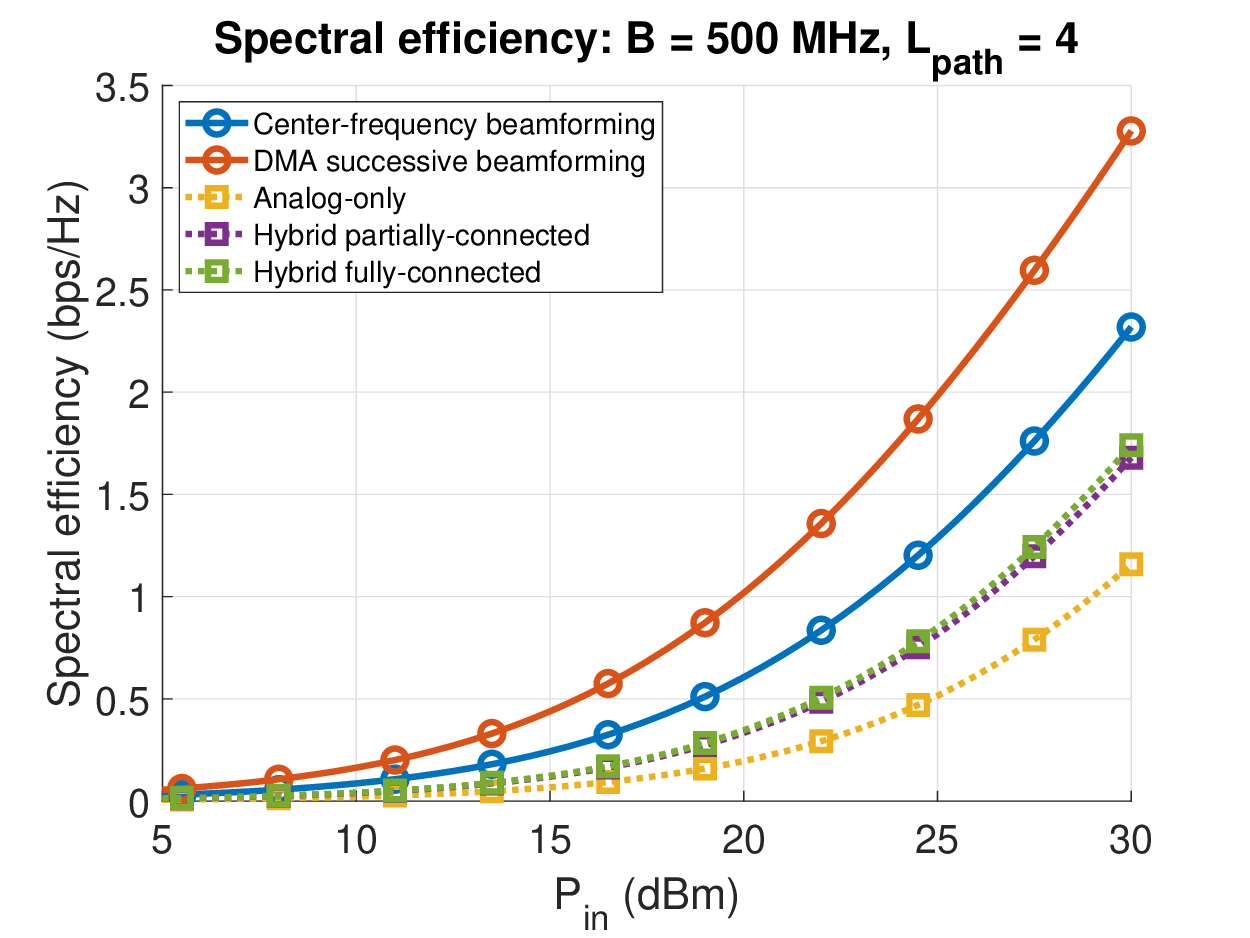}
    \caption{\blue{Spectral efficiency results for a multipath channel environment with $N_\mathsf{ele} = 6$ elements per subarray and $N_\mathsf{rf} = 4$ subarrays. We find that the proposed DMA successive beamforming algorithm outperforms both the baseline DMA center-frequency beamforming method and the hybrid architectures. This is due to the flexible algorithm design that adapts to the DMA and channel frequency-selectivity, dense DMA element spacing, and high component loss from passive phase shifters for the hybrid architecture.}}
    \label{fig: spec eff multipath}
\end{figure}

Fig. \ref{fig: spec eff multipath paths} shows simulation results for the average spectral efficiency as a function of the number of channel paths in the multipath environment. We assume the same array architectures for both the DMA and hybrid cases as in Fig. \ref{fig: spec eff multipath}, where $N_\mathsf{ele} = 6$ elements and $N_\mathsf{rf} = 4$ RF chains. We find that the proposed successive beamforming algorithm outperforms both the center-frequency algorithm and hybrid beamforming architectures in terms of spectral efficiency across every number of channel paths. Notably, as the number of channel paths increases, the gap between the DMA successive beamforming algorithm and the center-frequency beamforming algorithm also grows wider. This shows how the successive beamforming algorithm adapts better to a more complicated channel environment with multipath and additional frequency-selectivity. Similar to the results in Fig. \ref{fig: spec eff multipath}, the large performance difference between the DMA successive beamforming algorithm and the hybrid architecture lies mostly in the component loss for the integrated phase shifters and power dividers. Since the partially-connected architecture contains fewer power dividers than the fully-connected case, this is also why the partially-connected architecture outperforms the fully-connected architecture for the scenario with two paths. In summary, both Figs. \ref{fig: spec eff multipath} and \ref{fig: spec eff multipath paths} demonstrate the resilience of the proposed beamforming algorithm to complicated channel environments, showing that DMAs provide a promising alternative to hybrid arrays for wideband wireless communications.

% DMAs have the
% where $\Nt$ antenna elements are each connected to 

% We do a Monte Carlo simulation over 500 channel realizations to obtain the final spectral efficiency results.

\begin{figure}
    \centering
    \includegraphics[width=\figsizeii\linewidth]{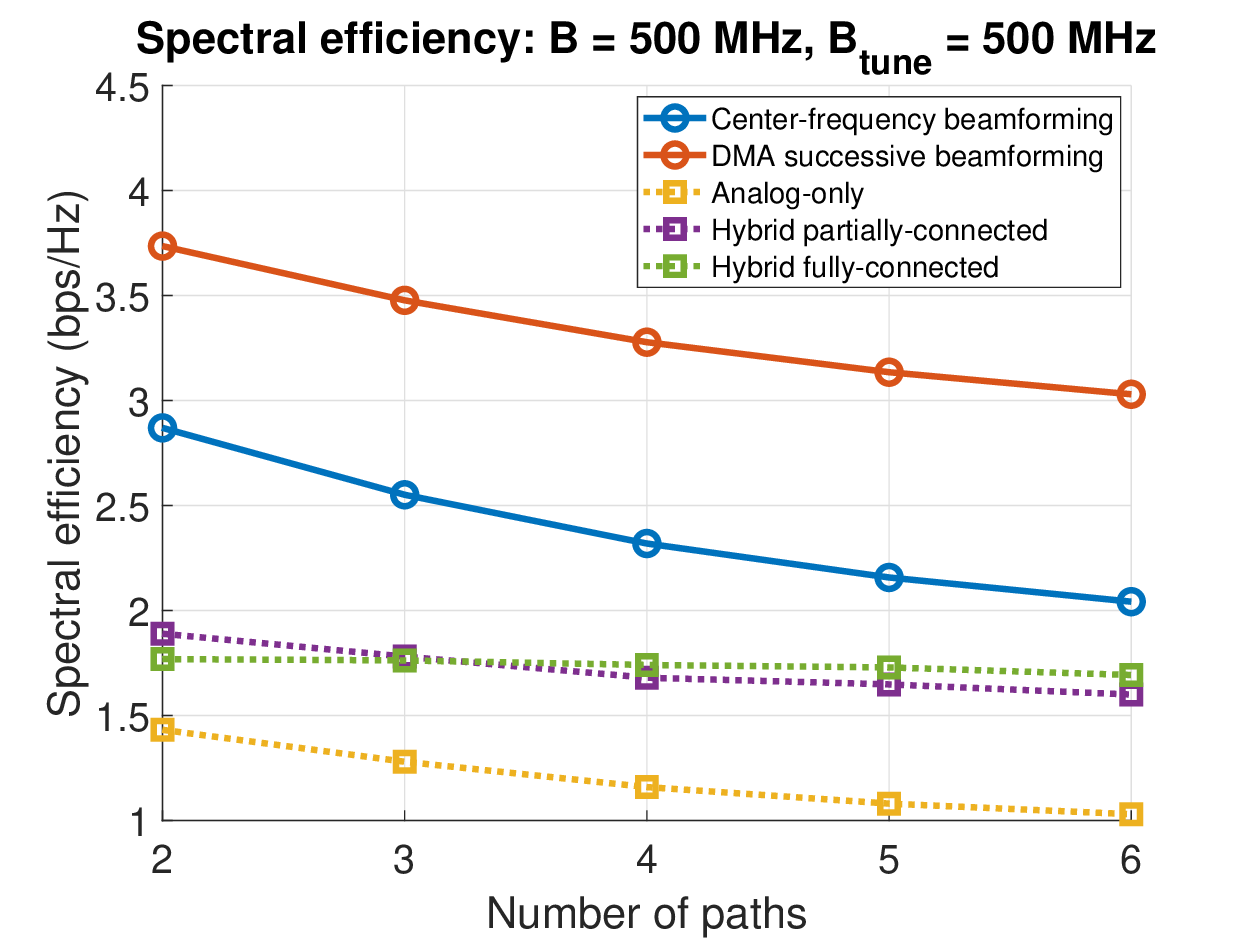}
    \caption{\blue{Average spectral efficiency  for a multipath channel environment as a function of the number of channel paths. We find that the proposed DMA successive beamforming algorithm provides the best performance across every number of channel paths.}}
    \label{fig: spec eff multipath paths}
\end{figure}

}

\section{Conclusions and future work}\label{sec: conclusion}
In this paper, we analyzed how DMA design characteristics impact DMA beamforming gain for large signal bandwidths. With the exception of very small signal bandwidths, we found that the frequency-selective DMA resonance, wireless channel, and waveguide channel effects cannot be neglected since the per-subcarrier beamforming gain degrades quickly outside the center frequency. We also showed that reducing the DMA element spacing can help decrease the amount of frequency-selectivity and improve the beamforming gain. Moreover, we illustrated how the tuning bandwidth plays a crucial role in maximizing the DMA performance by allowing for the entire set of Lorentzian-constrained weights to be available for beamforming. Lastly, we proposed a successive DMA beamforming algorithm that significantly improved the wideband DMA data rates by configuring the DMA elements to consider their frequency-selective response and the previously-configured DMA elements. We also demonstrated the effectiveness of the proposed algorithm in a multipath setting to show how the proposed algorithm can continue to improve spectral efficiency in complicated channel environments. A major limitation of our work is the that the DMA beamforming gain approximation was derived solely for a LOS channel. In future work, we plan to generalize the DMA approximation approach to non-LOS channels, which is non-trivial due to the additional frequency-selectivity and fading from a multi-path environment. We also plan to extend the successive beamforming algorithm to consider more complex optimization algorithms, which is challenging due to the non-convex nature of optimizing spectral efficiency under the Lorentzian-constrained weights.

\appendices

\section{Proof of Lemma~\ref{lem:  mag pol}}\label{proof:  lem mag pol}

	We derive here the frequency-selective Lorentzian constraint. We begin with the formulation of the magnetic polarizability in \eqref{eq: mag pol cos}. We omit the dependence of the magnetic polarizability $\tam$ and $\Psi$ on the subcarrier frequency and resonant frequency for brevity.  We then apply the identity $\cos(\Psi) = 1-2\sin^2(\frac{\Psi}{2})$ to rewrite \eqref{eq: mag pol cos} as 
	\begin{equation}\label{eq: mag pol lemma}
		\tam = \left( 1-2\sin^2\left(\frac{\Psi}{2}\right) \right)e^{\sfj(\Psi-\frac{\pi}{2})}.
	\end{equation}
	We then use the identity $\sin^2(\frac{\Psi}{2}) = \left( \frac{e^{\sfj \frac{\Psi}{2}}-e^{-\sfj \frac{\Psi}{2}}}{2\sfj} \right) ^2$ to obtain
	\begin{equation}
		\tam= -\frac{\sfj-e^{\sfj 2 \Psi}}{2}.
	\end{equation}

\section{Proof of Lemma~\ref{lem: 2}}\label{proof of lemma 2}

	We now define an approximation for the beamforming gain of a DMA. We omit the frequency-selectivity of the channel and DMA beamformer for brevity in this derivation. 
    To derive the beamforming gain $|\heff^{\mathsf{T}} \fdma|^2$, we find that we can split up the beamforming gain into two terms as
\begin{equation}
    |\heff^{\mathsf{T}} \fdma|^2 = \left|-\frac{\sfj}{2}\heff^{\mathsf{T}} \mathbf{1}+\frac{1}{2}\heff^{\mathsf{T}}e^{\sfj 2 \mathbf{\Psi}} \right|^2.
\end{equation}
For a wavenumber $k_0$, we define the propagation angle $\tg = \arcsin\left( \frac{\btg}{k_0} \right)$ as the propagation lobe created by the effective channel $-\frac{\sfj}{2}\heff^{\mathsf{T}} \mathbf{1}$, and the main lobe as $\frac{1}{2}\heff^{\mathsf{T}}e^{\sfj 2 \mathbf{\Psi}}$. We assume single-beam beamsteering to steer a main lobe in a desired direction $\phit$ for an LOS channel. We model the effective channel for the phase $\phio$ as $\heff = [e^{\sfj 0 \phio}, e^{ \sfj 1 \phio},\ldots, e^{\sfj (\Nt-1) \phio}]^T$. The effective channel phase is given by wireless channel and DMA channel phases as
	\begin{equation}\label{eq: effect channel phase}
		\phio= \dx \left( \frac{2 \pi \ft}{c}  \sin\phit + \btg \right).
	\end{equation}
	The normalized propagation lobe term $\Gp(\phit) = -\frac{\sfj}{2}\heff^{\mathsf{T}} \mathbf{1}$ can then be written as
	\begin{multline}\label{eq: prop lobe term}
		\Gp(\phit)  =\frac{1}{\Nt}\sum\limits_{\nth=1}^{\Nt} e^{\sfj (\nth-1) \phio} \\ = e^{\sfj \frac{(\Nt-1) \phio}{2}} \frac{1}{\Nt}\frac{\sin \frac{\Nt}{2} \phio}{\sin \frac{1}{2} \phio}.
	\end{multline}
	In the limit of a large number of elements as $\Nt \rightarrow \infty$, the magnitude of the propagation lobe term forms an infinitely thin sinc function and can be approximated as \cite{mack2007fundamental}
%	\begin{equation}
%		\lim\limits_{L \rightarrow \infty} |\Gp(\phit)| = \delta(\phit-\tg).
%	\end{equation}
	\begin{equation}
		\lim\limits_{\Nt \rightarrow \infty} |\Gp(\phit)| \approx \begin{cases}
			&1, \quad \text{if } \phit=\tg, \\
			&0, \quad \text{if } \phit\neq \tg.
		\end{cases}
	\end{equation}
	 Therefore, we find that for a sufficiently large $\Nt$ and $\phit\neq \tg$, the propagation lobe term $\Gp(\phit) \approx 0$ and
	\begin{equation}\label{eq: prop term approx}
		\left| \heff^{\mathsf{T}} \fdma \right|^2 \approx \left|\frac{1}{2}\heff^{\mathsf{T}}e^{\sfj 2 \mathbf{\Psi}} \right|^2.
	\end{equation}
    We can also avoid the propagation lobe altogether with the condition $\btg > k_0$ to ensure that the beamforming gain due to the propagation lobe is very small, since $\tg = \arcsin\left( \frac{\btg}{k_0} \right) \in \mathbb{C}$. For now, we assume that $\btg > k_0$ and $\Nt$ is sufficiently large such that \eqref{eq: prop term approx} holds.
	
\section{Proof of Lemma~\ref{lem: 3}}\label{proof: lem 3}

	We derive an approximation for the angle of the magnetic polarizability. We apply a Taylor series expansion to the polarizability function in \eqref{eq: mag pol} around $f=f_{\sfr, \nth}$. Let $\Psi(f)=\arctan \left( \frac{2\pi (f_{\sfr, \nth}^2-f^2)}{\Gamma f} \right) $. We obtain 
	\begin{align}
		\Psi(f) =\Psi(f_{\sfr, \nth})+  \Psi'(f_{\sfr, \nth})(f-f_{\sfr, \nth})+ O(f^2).
	\end{align}
	It can be verified that $\Psi'(f_{\sfr, \nth})=-\frac{4\pi }{\Gamma}$.

\section{Proof of Lemma~\ref{lem: 4}}\label{proof: lem4}
We determine an approximation for the impact of frequency selectivity on the DMA beamforming gain. Based on the channel phase in \eqref{eq: chin} and the linear DMA weight approximation in \eqref{eq: ang pol angle approx}, we can decompose the frequency-selective DMA beamforming gain as
        \begin{equation}
		\Fsel(f_k)   = \left| \frac{1}{2} \frac{1}{\sqrt{\Nt}}\sum\limits_{\nth=1}^{\Nt} e^{\sfj \frac{8\pi \Delta_k}{\Gamma}}e^{{\sfj}(\nth-1)\chi_{\sfo}[k]} \right|^2,
	\end{equation}
	which simplifies to
	\begin{equation}
		\Fsel(f_k)  = \left| \frac{1}{2}  \frac{1}{\sqrt{\Nt}} \frac{\sin\left( \frac{\Nt}{2}\chi_{\sfo}[k]\right)}{\sin \left( \frac{1}{2}\chi_{\sfo}[k] \right)}\right|^2.
	\end{equation}
	
	\section{Proof of Lemma~\ref{lem: 5}}\label{proof: lem5}

		We determine the impact of the tuning bandwidth and available DMA weights on the DMA beamforming gain. For a given weight fill ratio, we assume the center-frequency DMA beamformer and assume any channel phase outside the Lorentzian-constrained weight range maps to the outermost DMA weight, as this minimizes the phase error between the two weights. We assume there to be no waveguide attenuation.
        
        First, we define the phase term of the feasible DMA weights with a certain weight fill ratio as
		\begin{equation}
			\mathcal{W} = \left\{ e^{-\sfj\frac{\pi}{2}} e^{\sfj \delta} \; : \delta \in [-\xi,\xi] \right\}.
		\end{equation}
		Next, we assume that the channel weights are uniformly sampled on the complex unit circle such that $|h_\nth| = 1, \angle h_\nth \in U(0,2\pi]$. We justify this assumption based on the channel weight formulation. The channel weights sample the complex unit circle based on the array steering vector in \eqref{eq: steering vec} and the DMA waveguide channel in \eqref{eq: dma waveguide channel}. The phase advance for the channel $\phio$ in \eqref{eq: effect channel phase} is nonzero for $\phit \neq \tg$, meaning there will always be a phase advance between DMA elements as $\phio \neq 0$. As the number of elements $\Nt$ increases, the number of samples increases and begins to approach a uniform sampling of the complex unit circle. Then, the proportion of channel weights $h_\nth$ within the feasible DMA weights $\mathcal{W}$ is approximately equal to the normalized arc length of $\mathcal{W}$ as $\frac{\xi}{\pi}$. 
		
		Given this uniform sampling assumption, we define the beamforming gain for the DMA beamformer in three different regions of interest. Since losses due to the amplitude distribution of the DMA weights is incorporated in \eqref{eq: freq select}, we assume a unit-amplitude DMA beamformer given by $\wn \in \mathcal{W}$ to isolate the effects of the phase mismatch between the DMA beamformer and channel.
		\begin{itemize}
			\item $h_\nth \in \mathcal{W} \;$: Here, the effective channel phase lies within the DMA feasible weights, allowing $h_\nth \wn= 1$.
			\item $h_\nth  \not\in \mathcal{W} \;$ and $\mathsf{Re}\{ h_\nth \} \geq 0 \;$: The effective channel phase lies outside the DMA feasible weights leading to a phase mismatch. The effective channel phase will configure the DMA element to map to the outermost weight in the righthand side of the complex plane. The beamforming gain is then $h_\nth \wn= h_\nth e^{-\sfj\frac{\pi}{2}} e^{\sfj \xi}$.
			\item $h_\nth  \not\in \mathcal{W} \;$ and $\mathsf{Re}\{ h_\nth \} < 0 \;$: The effective channel phase lies outside the DMA feasible weights and in the lefthand side of the complex plane. Similar to the prior region, the beamforming gain is then $h_\nth \wn= h_\nth e^{-\sfj\frac{\pi}{2}} e^{-\sfj \xi}$.
		\end{itemize}
		We take the average value of the beamforming gain normalized by the number of elements $\Nt$ for these three regions as the beamforming gain term due to an incomplete weight fill, defined by
        \begin{multline}
			\Wfill(\xi) =   \Biggl| \frac{1}{2\pi} \Biggl\{ \int_{-\frac{\pi}{2}-\xi}^{-\frac{\pi}{2}+\xi}  1dz + \int_{\frac{\pi}{2}+\xi}^{\frac{\pi}{2}} e^{-{\mathsf{j}} \frac{\pi}{2} }  e^{{\mathsf{j}} \xi } e^{-{\mathsf{j}} z }dz  + \\ \int_{-\frac{\pi}{2}-\xi}^{\frac{\pi}{2}} e^{-{\mathsf{j}} \frac{\pi}{2} }  e^{-{\mathsf{j}} \xi } e^{-{\mathsf{j}} z }dz  \Biggr\} \Biggr|^2  = \left| \frac{1}{2\pi} \left( 2\sin (\xi) + 2\xi \right) \right|^2.
		\end{multline}

\section{Proof of Lemma~\ref{lem: 6}}\label{proof: lem6}

We derive an expression for the impact of the waveguide field attenuation on the resulting DMA beamforming gain. The penalty term  $\Aleak(\Lambda)$  from \eqref{eq: bf gain atten} is simplified as
	\begin{align}
		\Aleak(\Lambda) =  \frac{\left| \sum\limits_{\nth=1}^{\Nt} e^{-\atg [k]\dx (\nth-1)}\right|^2}{\Nt\sum\limits_{\nth=1}^{\Nt} \left|e^{-\atg [k]\dx (\nth-1)} \right|^2 } .
	\end{align}
	We find that both the numerator and denominator have the form of a geometric series that is simplified  as
	\begin{equation}\label{eq: S tanh}
		\Aleak(\Lambda) = \frac{\left| \frac{1-e^{-\atg \dx (\Nt-1) }}{1-e^{-\atg \dx }} \right|^2}{\Nt\frac{1-e^{-2\atg \dx (\Nt-1) }}{1-e^{-2\atg \dx }}} = \frac{\tanh \left(\frac{-\atg (\Nt-1) \dx}{2} \right)}{\Nt\tanh(\frac{-\atg \dx}{2})}.
	\end{equation}
	Lastly, we have defined the waveguide attenuation constant $\atg$ to ensure that a specified fraction of radiated power $\Lambda$ is met as $e^{-2\atg(\Nt-1)\dx} = 1-\Lambda$. We can then solve for $\atg$ as $ \atg = -\frac{\ln(1-\Lambda)}{2(\Nt-1)\dx}$. Plugging this back into \eqref{eq: S tanh} gives
	\begin{equation}
		\Aleak(\Lambda) = \frac{\tanh \left(\frac{\ln(1-\Lambda)}{4} \right)}{\Nt\tanh \left(\frac{\ln(1-\Lambda)}{4(\Nt-1)} \right)}.
	\end{equation}
	As $\Nt$ grows larger, we make the approximation that $ \tanh(x) = x + O(x^2), \; x << 1$ and $\frac{\Nt-1}{\Nt} \approx 1$ to obtain the final expression
	\begin{align}\label{eq: approx atten}
		\Aleak(\Lambda)  \approx   \frac{4}{\ln(1-\Lambda)}  \tanh\left(\frac{\ln(1-\Lambda)}{4} \right) .
	\end{align}

\bibliographystyle{IEEEtran}

\bibliography{myrefs_JC}

\end{document}